\documentclass[runningheads]{llncs}

\usepackage[T1]{fontenc}
\usepackage{graphicx}
\usepackage{amsmath}
\usepackage{amssymb}
\usepackage{hyperref}
\usepackage{marvosym}
\usepackage{booktabs}
\usepackage{cite}
\usepackage{afterpage}

\usepackage{color}

\newcommand{\defref}[1]{Definition~\ref{def:#1}}
\newcommand{\secref}[1]{Section~\ref{sec:#1}}
\newcommand{\exaref}[1]{Example~\ref{exa:#1}}
\newcommand{\figref}[1]{Figure~\ref{fig:#1}}
\newcommand{\tabref}[1]{Table~\ref{tab:#1}}
\newcommand{\lemref}[1]{Lemma~\ref{lem:#1}}
\newcommand{\thmref}[1]{Theorem~\ref{thm:#1}}
\newcommand{\corref}[1]{Corollary~\ref{cor:#1}}

\renewcommand{\geq}{\geqslant}
\renewcommand{\leq}{\leqslant}
\newcommand{\limp}{\Rightarrow}

\newcommand{\dpi}{\mathrm{\Pi}}
\newcommand{\tp}{\mathsf{tp}}
\newcommand{\per}{\mathsf{per}}
\newcommand{\thy}{\mathsf{Thy}}
\newcommand{\ctx}{\mathsf{Ctx}}

\newcommand{\ds}{\displaystyle}

\newcommand{\xT}{\mathcal{T}}

\newcommand{\xV}{\mathcal{V}}

\newcommand{\negr}{\xT_{\lnot}}
\newcommand{\neqr}{\xT_{\neq}}
\newcommand{\dnegr}{\xT_{\lnot\lnot}}
\newcommand{\impr}{\xT_{\limp}}
\newcommand{\nimpr}{\xT_{\lnot\limp}}
\newcommand{\symcastar}{\xT_{\mathrm{SYMCAST}_1}}
\newcommand{\symcastbr}{\xT_{\mathrm{SYMCAST}_2}}
\newcommand{\allr}{\xT_{\forall}}
\newcommand{\nallr}{\xT_{\lnot\forall}}
\newcommand{\ber}{\xT_{\mathrm{BE}}}
\newcommand{\bqr}{\xT_{\mathrm{BQ}}}
\newcommand{\fer}{\xT_{\mathrm{FE}}}
\newcommand{\fqr}{\xT_{\mathrm{FQ}}}
\newcommand{\matr}{\xT_{\mathrm{MAT}}}
\newcommand{\decr}{\xT_{\mathrm{DEC}}}
\newcommand{\conr}{\xT_{\mathrm{CON}}}
\newcommand{\erar}{\xT_{\mathrm{ER}_1}}
\newcommand{\erbr}{\xT_{\mathrm{ER}_2}}

\newcommand{\elem}{\mathsf{elem}}
\newcommand{\nat}{\mathsf{nat}}
\newcommand{\lst}{\mathsf{lst}}
\newcommand{\zro}{\mathsf{0}}
\newcommand{\suc}{\mathsf{s}}
\newcommand{\nil}{\mathsf{nil}}
\newcommand{\cons}{\mathsf{cons}}
\newcommand{\plus}{\mathsf{plus}}
\newcommand{\app}{\mathsf{app}}
\newcommand{\rev}{\mathsf{rev}}

\begin{document}

\title{Tableaux for Automated Reasoning in Dependently-Typed Higher-Order Logic}
\titlerunning{Tableaux for Automated Reasoning in DHOL}

\author{Johannes Niederhauser\inst{1}\textsuperscript{(\Letter)}
  \orcidID{0000-0002-8662-6834} \and
  Chad E.~Brown\inst{2} \and
  Cezary Kaliszyk\inst{1,3}\orcidID{0000-0002-8273-6059}}

\authorrunning{J. Niederhauser et al.}
\institute{
  Department of Computer Science, University of Innsbruck,
  Innsbruck, Austria
  \email{johannes.niederhauser@uibk.ac.at}
  \and
  Czech Institute of Informatics, Robotics and Cybernetics,
  Czech Technical University in Prague, Prague, Czech Republic
  \and
  School of Computing and Information Systems, University of Melbourne,
  Melbourne, Australia \\
  \email{cezarykaliszyk@gmail.com}
}

\maketitle

\begin{abstract}
  Dependent type theory gives an expressive type system
  facilitating succinct formalizations of mathematical concepts. In
  practice, it is mainly used for interactive theorem proving with
  intensional type theories, with PVS being a notable exception. In
  this paper, we present native rules for automated reasoning in a
  dependently-typed version (DHOL) of classical higher-order logic
  (HOL). DHOL has an extensional type theory with an undecidable type
  checking problem which contains theorem proving. We implemented the
  inference rules as well as an automatic type checking mode in Lash,
  a fork of Satallax, the leading tableaux-based prover for HOL. Our
  method is sound and complete with respect to provability in
  DHOL. Completeness is guaranteed by the incorporation of a sound and
  complete translation from DHOL to HOL recently proposed by Rothgang
  et al. While this translation can already be used as a preprocessing
  step to any HOL prover, to achieve better performance, our system
  directly works in DHOL. Moreover, experimental results show that the
  DHOL version of Lash can outperform all major HOL provers executed
  on the translation.
  \keywords{Tableaux \and Dependent Types \and Higher-Order Logic}
\end{abstract}

\section{Introduction}
\label{sec:intro}

Dependent types introduce the powerful concept of types depending on
terms. Lists of fixed length are an easy but interesting
example. Instead of having a simple type $\lst$ we may have a type
$\dpi n\colon\nat.\: \lst\: n$ which takes a natural number as
argument and returns the type of a list with length $n$. More
generally, lambda terms $\lambda x.s$ now have a dependent type
$\dpi x\colon A. B$ which makes the type of $(\lambda x.s)\: t$
dependent on $t$. With that, it is possible for example to specify an
unfailing version of the tail function by declaring its type to be
$\dpi n\colon\nat.\: \lst\: (\suc\: n) \to \lst\: n$.  Many
interactive theorem provers for dependent type theory are available
\cite{COQ,dMU21,BDN09,PS99}, most of them implement intensional type
theories, i.e., they distinguish between a decidable judgmental
equality (given by conversions) and provable equality (inhabiting an
identity type).  Notable exceptions are PVS\cite{ROS98} and
F$^{\star}$\cite{S+16} which implement an extensional type theory. In
the context of this paper, we say a type theory is \emph{extensional}
if judgmental equality and provable equality coincide, as
in~\cite{ML84}. The typing judgment in such type theories is usually
undecidable, as shown in~\cite{CCD17}.

The broader topic of this paper is automated reasoning support for
extensional type theories with dependent types. Not much has been done
to this end, but last year Rothgang et al.~\cite{RRB23} introduced an
extension of HOL to dependent types which they dub DHOL. In contrast
to dependent type theory, automated theorem proving in HOL has a long
history and led to the development of sophisticated provers
\cite{BBTV23,SB21,B12}. Rothgang et al.~defined a natural extension of
HOL and equipped it with automation support by providing a sound and
complete translation from DHOL into HOL. Their translation has been
implemented and can be used as a preprocessing step to any HOL prover
in order to obtain an automated theorem prover for DHOL.  Hence, by
committing to DHOL, automated reasoning support for extensional
dependent type theories does not have to be invented from scratch but
can benefit from the achievements of the automated theorem proving
community for HOL.

In this paper, we build on top of the translation from Rothgang et
al.~to develop a tableau calculus which is sound and complete for
DHOL. In addition, dedicated inference rules for DHOL are defined and
their soundness is proved. The tableau calculus is implemented as an
extension of Lash \cite{BK22}. The remainder of this paper is
structured as follows: \secref{preliminaries} sets the stage by
defining DHOL and the erasure from DHOL to HOL due to Rothgang et
al. before \secref{calculus} defines the tableau calculus and provides
soundness and completeness proofs. The implementation is described in
\secref{implementation}. Finally, we report on experimental results in
\secref{experiments}.

\section{Preliminaries}
\label{sec:preliminaries}

\subsection{HOL}

We start by giving the syntax of higher-order logic (HOL) which goes
back to Church \cite{C40}. In order to allow for a graceful extension
to DHOL, we define it with a grammar based on \cite{RRB23}.
\begin{align*}
  T &::= \circ \:\mid\: T,a\colon\tp \:\mid\: T,x\colon A \:\mid\: T,s \label{eq:holtheories} \tag{theories} \\
  \Gamma &::= \cdot \:\mid\: \Gamma,x\colon A \:\mid\: \Gamma,s \label{eq:holcontexts} \tag{contexts} \\
  A,B &::= a \:\mid\: A \to B \:\mid\: o \label{eq:holtypes} \tag{types} \\
  s,t,u,v &::= x \:\mid\: \lambda x\colon A.s \:\mid\: s\: t \:\mid\: \bot \:\mid\: \lnot s \:\mid\: s \limp t \:\mid\: s =_{A} t \:\mid\: \forall x\colon A.s \label{eq:holterms} \tag{terms}
\end{align*}

A theory consists of base type declarations $a\colon\tp$, typed
variable or constant declarations $x\colon A$ and axioms. Contexts are
like theories but without base type declarations. In the following, we
will often write $s \in T,\Gamma$ to denote that $s$ occurs in the
combination of $T$ and $\Gamma$. Furthermore, note that $\circ$ and
$\cdot$ denote the empty theory and context, respectively. Types are
declared base types $a$, the base type of booleans $o$ or function
types $A \to B$. As usual, the binary type constructor $\to$ is
right-associative. Terms are simply-typed lambda-terms (modulo
$\alpha$-conversion) enriched by the connectives $\bot$, $\lnot$,
$\limp$, $=_{A}$ as well as the typed binding operator for
$\forall$. All connectives yield terms of type $o$ (formulas). By
convention, application associates to the left, so $s\: t\: u$ means
$(s\: t)\: u$ with the exception that $\lnot\: s\: t$ always means
$\lnot(s\: t)$. Moreover, we abbreviate $\lnot(s =_{A} t)$ by
$s \neq_{A} t$ and sometimes omit the type subscript of $=_{A}$ when
it is either clear from the context or irrelevant. We write
$s[x_1/t_1,\dots,x_n/t_n]$ to denote the simultaneous capture-avoiding
substitution of the $x_i$'s by the $t_i$'s. The set of free variables
of a term $s$ is denoted by $\xV s$.

A theory $T$ is well-formed if all types are well-formed and axioms
have type $o$ with respect to its base type declarations. In that
case, we write $\vdash^{\mathsf{s}} T\:\thy$ where the superscript
$\mathsf{s}$ indicates that we are in the realm of simple types. Given
a well-formed theory $T$, the well-formedness of a context $\Gamma$ is
defined in the same way and denoted by
$\vdash^{\mathsf{s}}_T \Gamma\:\ctx$.  Given a theory $T$ and a
context $\Gamma$, we write $\Gamma \vdash^{\mathsf{s}}_T A\:\tp$ to
state that $A$ is a well-formed type and
$\Gamma \vdash^{\mathsf{s}}_T s\colon A$ to say that $s$ has type $A$.
Furthermore, $\Gamma \vdash^{\mathsf{s}}_{T} s$ denotes that $s$ has
type $o$ and is provable from $\Gamma$ and $T$ in HOL. Finally, we use
$\Gamma \vdash^{\mathsf{s}}_T A \equiv B$ to state that $A$ and $B$
are equivalent well-formed types. For HOL this is trivial as it
corresponds to syntactic equivalence, but this will change drastically
in DHOL.

\subsection{DHOL}

The extension from HOL to DHOL consists of two crucial ingredients:
\begin{itemize}
\item the type constructor $A \to B$ is replaced by the constructor
  $\dpi x\colon A. B$ which potentially makes the return type $B$
  dependent on the actual argument $x$; we stick to the usual arrow
  notation if $B$ does not contain $x$
\item base types $a$ can now take term arguments; for an n-ary base
  type we write
  $a\colon\dpi x_1\colon A_1.\:\cdots\: \dpi x_n\colon A_n.\:\tp$
\end{itemize}
Thus, the grammar defining the syntax of DHOL is given as follows:
\begin{align*}
  T &::= \circ \:\mid\: T,a\colon(\dpi x\colon A.)^*\:\tp \:\mid\: T,x\colon A \:\mid\: T,s \label{eq:dholtheories} \tag{theories} \\
  \Gamma &::= {\cdot} \:\mid\: \Gamma,x\colon A \:\mid\: \Gamma,s \label{eq:dholcontexts} \tag{contexts} \\
  A,B &::= a\:t_1\: \dots\: t_n \:\mid\: \dpi x\colon A.B \:\mid\: o \label{eq:dholtypes} \tag{types} \\
  s,t,u,v &::= x \:\mid\: \lambda x\colon A.s \:\mid\: s\: t \:\mid\: \bot \:\mid\: \lnot s \:\mid\: s \limp t \:\mid\: s =_{A} t \:\mid\: \forall x\colon A.s \label{eq:dholterms} \tag{terms}
\end{align*}
If a base type $a$ has arity 0, it is called a \emph{simple base
  type}. Note that HOL is the fragment of DHOL where all base types
have arity 0. Allowing base types to have term arguments makes type
equality a highly non-trivial problem in DHOL. For example, if
$\Gamma \vdash^{\mathsf{d}}_T s\colon \dpi x\colon A.B$ (the
$\mathsf{d}$ in $\vdash^{\mathsf{d}}$ indicates that we are speaking
about DHOL) and $\Gamma \vdash^{\mathsf{d}}_T t \colon A'$ we still
want $\Gamma \vdash^{\mathsf{d}}_T (s\: t)\colon B[x/t]$ to hold if
$\Gamma \vdash^{\mathsf{d}}_T A \equiv A'$, so checking whether two
types are equal is a problem which occurs frequently in
DHOL. Intuitively, we have $\Gamma \vdash^{\mathsf{d}}_T A \equiv A'$
if and only if their simply-typed skeleton consisting of arrows and
base types without their arguments is equal and given a base type
$a\colon\dpi x_1\colon A_1.\:\cdots\: \dpi x_n\colon A_n.\:\tp$, an
occurrence $a\: t_1\: \dots\: t_n$ in $A$ and its corresponding
occurrence $a\: t_1'\: \dots\: t_n'$ in $A'$, we have
$\Gamma \vdash^{\mathsf{d}}_T t_i
=_{A_i[x_1/t_1,\dots,x_{i-1}/t_{i-1}]} t_i'$ for all
$1 \leq i \leq n$. This makes DHOL an extensional type theory where
already type checking is undecidable as it requires theorem
proving. Another difference from HOL is the importance of the chosen
representation of contexts and theories: Since the well-typedness of a
term may depend on other assumptions, the order of the type
declarations and formulas in a context $\Gamma$ or theory $T$ is
relevant. A formal definition of the judgments
$\Gamma \vdash^{\mathsf{d}}_T A\:\tp$,
$\Gamma \vdash^{\mathsf{d}}_T s\colon A$,
$\Gamma \vdash^{\mathsf{d}}_T s$ and
$\Gamma \vdash^{\mathsf{d}}_T A \equiv B$ via an inference system is
given in \cite{RRB23}. Since we use more primitive connectives, a
minor variant is presented in \figref{dholndcalculus}.

\begin{figure}[t]
  \centering
  \begin{gather*}
    \frac{}{\vdash^{\mathsf{d}} {\circ}\:\thy} \text{thyEmpty}
    \qquad
    \frac{\vdash^{\mathsf{d}}_T x_1\colon A_1,\dots,x_n\colon A_n\:\ctx}{\vdash^{\mathsf{d}} T, a\colon \dpi x_1\colon A_1.\: \dots\: \dpi x_n\colon A_n.\tp\:\thy} \text{thyType}
    \\
    \frac{\vdash^{\mathsf{d}}_T A\:\tp}{\vdash^{\mathsf{d}} T, x\colon A\:\thy} \text{thyConst}
    \qquad
    \frac{\vdash^{\mathsf{d}}_T s\colon o}{\vdash^{\mathsf{d}} T, s\:\thy} \text{thyAxiom}
    \qquad
    \frac{\vdash^{\mathsf{d}} T\:\thy}{\vdash^{\mathsf{d}}_T {\cdot}\:\ctx} \text{ctxEmpty}
    \\
    \frac{\Gamma \vdash^{\mathsf{d}}_T A\:\tp}{\vdash^{\mathsf{d}}_T \Gamma, x\colon A\:\ctx} \text{ctxVar}
    \qquad
    \frac{\Gamma \vdash^{\mathsf{d}}_T s\colon o}{\vdash^{\mathsf{d}}_T \Gamma, s\:\ctx} \text{ctxAssume}
    \\
    \frac{a\colon (\cdots\dpi x_i\colon A_i.\cdots.\tp) \in T \:\: \vdash^{\mathsf{d}}_T \Gamma\:\ctx \:\: \cdots \Gamma \vdash^{\mathsf{d}}_T t_i\colon A_i[x_1/s_1\dots x_{i-1}/s_{i-1}] \cdots}{\Gamma \vdash^{\mathsf{d}}_T a\: t_1 \:\dots\: t_n \:\tp} \text{type}
    \\
    \frac{x\colon A' \in T \quad \Gamma \vdash^{\mathsf{d}}_T A' \equiv A}{\Gamma \vdash^{\mathsf{d}}_T x\colon A} \text{const}
    \qquad
    \frac{s \in T \quad \vdash^{\mathsf{d}}_T \Gamma\:\ctx}{\Gamma \vdash^{\mathsf{d}}_T s} \text{axiom}
    \\
    \frac{x\colon A' \in \Gamma \quad \Gamma \vdash^{\mathsf{d}}_T A' \equiv A}{\Gamma \vdash^{\mathsf{d}}_T x\colon A} \text{var}
    \qquad
    \frac{s \in \Gamma \quad \vdash^{\mathsf{d}}_T \Gamma\:\ctx}{\Gamma \vdash^{\mathsf{d}}_T s} \text{assume}
    \qquad
    \frac{\vdash^{\mathsf{d}}_T \Gamma\:\ctx}{\Gamma \vdash^{\mathsf{d}}_T o\:\tp} \text{bool}
    \\
    \frac{\Gamma \vdash^{\mathsf{d}}_T A\:\tp \quad \Gamma,x\colon A \vdash^{\mathsf{d}}_T B\:\tp}
         {\Gamma \vdash^{\mathsf{d}}_T \dpi x\colon A.B\:\tp} \text{pi}
    \qquad
    \frac{\Gamma \vdash^{\mathsf{d}}_T A \equiv A' \quad \Gamma, x\colon A \vdash^{\mathsf{d}}_T B \equiv B'}
         {\Gamma \vdash^{\mathsf{d}}_T \dpi x\colon A.B \equiv \dpi x\colon A'.B'} \text{cong}\dpi
    \\
    \frac{a\colon (\cdots\dpi x_i\colon A_i.\cdots.\tp) \in T \:\: \vdash^{\mathsf{d}}_T \Gamma\:\ctx \:\: \cdots \Gamma \vdash^{\mathsf{d}}_T s_i =_{A_i[x_1/s_1\dots x_{i-1}/s_{i-1}]} t_i \cdots}
    {\Gamma \vdash^{\mathsf{d}}_T a\: s_1 \:\dots\: s_n \equiv a\: t_1 \:\dots\: t_n} \text{congBT}
    \\
    \frac{\Gamma, x\colon A \vdash^{\mathsf{d}}_T t\colon B}{\Gamma \vdash^{\mathsf{d}}_T (\lambda x\colon A.t)\colon \dpi x\colon A.B} \text{lambda}
    \qquad
    \frac{\Gamma \vdash^{\mathsf{d}}_T s\colon \dpi x\colon A.B \quad \Gamma \vdash^{\mathsf{d}}_T t\colon A}
         {\Gamma \vdash^{\mathsf{d}}_T (s\: t)\colon B[x/t]} \text{appl}
    \\
    \frac{\Gamma \vdash^{\mathsf{d}}_T s\colon A \quad \Gamma \vdash^{\mathsf{d}}_T t\colon A}{\Gamma \vdash^{\mathsf{d}}_T (s =_A t)\colon o} \text{=type}
    \qquad
    \frac{\Gamma \vdash^{\mathsf{d}}_T A \equiv A' \quad \Gamma, x\colon A \vdash^{\mathsf{d}}_T t =_B t'}
         {\Gamma \vdash^{\mathsf{d}}_T \lambda x\colon A. t =_{\dpi x\colon A.B} \lambda x\colon A'.t'} \text{cong}\lambda
    \\
    \frac{\Gamma \vdash^{\mathsf{d}}_T s =_{\dpi x\colon A.B} s' \quad \Gamma \vdash^{\mathsf{d}}_T t =_A t'}
         {\Gamma \vdash^{\mathsf{d}}_T s\: t =_{B[x/t]} s'\: t'} \text{congAppl}
    \qquad
    \frac{\Gamma \vdash^{\mathsf{d}}_T s\colon A}{\Gamma \vdash^{\mathsf{d}}_T s =_A s} \text{refl}
    \\
    \frac{\Gamma \vdash^{\mathsf{d}}_T s =_A t}{\Gamma \vdash^{\mathsf{d}}_T t =_A s} \text{sym}
    \qquad
    \frac{\Gamma \vdash^{\mathsf{d}}_T s\colon (\dpi x\colon A.B) \quad x \not\in \xV s}
         {\Gamma \vdash^{\mathsf{d}}_T s =_{\dpi x\colon A.B} \lambda x\colon A. s\: x} \text{eta}
    \\
    \frac{\Gamma \vdash^{\mathsf{d}}_T (\lambda x\colon A.s)\: t\colon B}
         {\Gamma \vdash^{\mathsf{d}}_T (\lambda x\colon A.s)\: t =_B s[x/t]} \text{beta}
    \qquad
    \frac{\vdash^{\mathsf{d}}_T \Gamma\:\ctx}
         {\Gamma \vdash^{\mathsf{d}}_T \bot\colon o} \bot\text{type}
    \qquad
    \frac{\Gamma \vdash^{\mathsf{d}}_T s\colon o \quad \Gamma \vdash^{\mathsf{d}}_T \bot}
         {\Gamma \vdash^{\mathsf{d}}_T s} \bot\text{e}
    \\
    \frac{\Gamma \vdash^{\mathsf{d}}_T s\colon o}
         {\Gamma \vdash^{\mathsf{d}}_T (\lnot s)\colon o} \lnot\text{type}
    \qquad
    \frac{\Gamma \vdash^{\mathsf{d}}_T s\colon o \quad \Gamma, s \vdash^{\mathsf{d}}_T \bot}
         {\Gamma \vdash^{\mathsf{d}}_T \lnot s} \lnot\text{i}
    \qquad
    \frac{\Gamma \vdash^{\mathsf{d}}_T s \quad \Gamma \vdash^{\mathsf{d}}_T \lnot s}
         {\Gamma \vdash^{\mathsf{d}}_T \bot} \lnot\text{e}
    \\
    \frac{\Gamma \vdash^{\mathsf{d}}_T \lnot \lnot s}
         {\Gamma \vdash^{\mathsf{d}}_T s} \lnot\lnot\text{e}
    \qquad     
    \frac{\Gamma \vdash^{\mathsf{d}}_T s\colon o \quad \Gamma,s \vdash^{\mathsf{d}}_T t\colon o}
         {\Gamma \vdash^{\mathsf{d}}_T (s \limp t)\colon o} {\limp}\text{type}
    \qquad
    \frac{\Gamma \vdash^{\mathsf{d}}_T s\colon o \quad \Gamma,s \vdash^{\mathsf{d}}_T t}
         {\Gamma \vdash^{\mathsf{d}}_T s \limp t} {\limp}\text{i}
    \\
    \frac{\Gamma \vdash^{\mathsf{d}}_T s \limp t \quad \Gamma \vdash^{\mathsf{d}}_T s}
         {\Gamma \vdash^{\mathsf{d}}_T t} {\limp}\text{e}
    \qquad
    \frac{\Gamma, x\colon A \vdash^{\mathsf{d}}_T s\colon o}
         {\Gamma \vdash^{\mathsf{d}}_T \forall x\colon A.s\colon o} \forall\text{type}
    \qquad
    \frac{\Gamma, x \colon A \vdash^{\mathsf{d}}_T s}
         {\Gamma \vdash^{\mathsf{d}}_T \forall x\colon A.s} \forall\text{i}
    \\
    \frac{\Gamma \vdash^{\mathsf{d}}_T \forall x\colon A.s \quad \Gamma \vdash^{\mathsf{d}}_T t\colon A}
         {\Gamma \vdash^{\mathsf{d}}_T s[x/t]} \forall\text{e}
    \qquad
    \frac{\Gamma \vdash^{\mathsf{d}}_T A \equiv A' \quad \Gamma, x\colon A \vdash^{\mathsf{d}}_T s =_o s'}
         {\Gamma \vdash^{\mathsf{d}}_T \forall x\colon A.s =_o \forall x\colon A'.s'} \text{cong}\forall
    \\
    \frac{\Gamma \vdash^{\mathsf{d}}_T s =_o s' \quad \Gamma \vdash^{\mathsf{d}}_T s'}{\Gamma \vdash^{\mathsf{d}}_T s} \text{cong}{\vdash}
    \qquad
    \frac{\Gamma \vdash^{\mathsf{d}}_T s \bot \quad \Gamma \vdash^{\mathsf{d}}_T s (\lnot \bot)}{\Gamma \vdash^{\mathsf{d}}_T \forall x\colon o. sx} \text{boolExt}
    \\
    \frac{\Gamma \vdash^{\mathsf{d}}_T s\colon o \quad \Gamma,x\colon A \vdash^{\mathsf{d}}_T s \quad A \text{ simple type}}
         {\Gamma \vdash^{\mathsf{d}}_T s} \text{nonempty}
  \end{gather*}
  \caption{Natural Deduction Calculus for DHOL}
  \label{fig:dholndcalculus}
\end{figure}

\afterpage{\clearpage}

\begin{example}
  \label{exa:plus_app}
  Consider the simple base types $\nat\colon \tp$ and
  $\elem\colon \tp$ as well as the dependent base type
  $\lst\colon \dpi x\colon \nat.\:\tp$. The constants and functions
  \begin{align*}
    \zro\colon& \nat &  \suc\colon& \nat \to \nat \\
    \nil\colon& \lst\:\zro & \cons\colon & \dpi n\colon \nat.\: \elem \to \lst\: n \to \lst\: (\suc\: n)
  \end{align*}
  provide means to represent their inhabitants. Additionally, we
  define functions $\plus\colon \nat \to \nat \to \nat$
  \begin{align*}
    \forall n\colon\nat.\: \plus\:\zro\: n &=_{\nat} n &
    \forall n,m\colon\nat.\: \plus\: (\suc\: n)\: m &=_{\nat} \suc\:(\plus\: n\: m)
  \end{align*}
  and $\app\colon \dpi n\colon\nat.\:\dpi m\colon\nat.\:\lst\:n \to \lst\:m \to \lst\: (\plus\:n\:m)$:
  \begin{align*}
    & \forall n\colon\nat,x\colon\lst\: n.\: \app\:\zro\: n\:\nil\: x =_{\lst\: n} x \\
    & \forall n,m\colon\nat,z\colon\elem,x\colon\lst\: n,y\colon\lst\: m. \\
    & \hspace{.6cm}\app\: (\suc\: n)\: m\: (\cons\: n\: z\: x)\: y =_{\lst\: (\suc\: (\plus\: n\: m))} \cons\: (\plus\: n\: m)\: z\: (\app\: n\: m\: x\: y)
  \end{align*}
  In the defining equations of $\app$, we annotated the equality sign
  with the dependent type of the term on the right-hand side. In all
  cases, the simply-typed skeleton is just $\lst$ but for a type check
  we need to prove the two equalities
  \begin{align*}
    \forall n\colon\nat.\: \plus\:\zro\: n &=_{\nat} n &
    \forall n,m\colon\nat.\: \plus\: (\suc\: n)\: m &=_{\nat} \suc\: (\plus\: n\: m)
  \end{align*}
  which are exactly the corresponding axioms for $\plus$. Type
  checking the conjecture
  \[\forall n\colon\nat, x\colon\lst\: n.\: \app\: n\: \zro\: x\: \nil =_{\lst\: n} x \]
  would require proving
  $\forall n\colon\nat.\:\plus\,\, n\,\, \zro =_{\nat} n$ which can be
  achieved by induction on natural numbers if we include the Peano
  axioms.
\end{example}

\subsection{Erasure}

The following definition presents the translation from DHOL to HOL due
to Rothgang et al.~\cite{RRB23}. Intuitively, the translation erases
dependent types to their simply typed skeletons by ignoring arguments
of base types. The thereby lost information on concrete base type
arguments is restored with the help of a partial equivalence relation
(PER) $A^*$ for each type $A$. A PER is a symmetric, transitive
relation. The elements on which it is also reflexive are intended to
be the members of the original dependent type, i.e.,
$\Gamma \vdash^{\mathsf{d}}_T s\colon A$ if and only if
$\overline{\Gamma} \vdash^{\mathsf{s}}_{\overline{T}} A^*\:
\overline{s}\: \overline{s}$.

\begin{definition}
  The translation from DHOL to HOL is given by the erasure function
  $\overline{s}$ as well as $A^*$ which computes the formula
  representing the corresponding PER of a type $A$. The functions are
  mutually defined by recursion on the grammar of DHOL.  The erasure
  of a theory (context) is defined as the theory (context) which
  consists of its erased components.
  \begin{align*}
    \overline{o} \:&=\: o &
    \overline{a\: t_1 \:\dots\: t_n} \:&=\: a \\
    \overline{\dpi x\colon A.B} \:&=\: \overline{A} \to \overline{B} &
    \overline{x} \:&=\: x \\
    \overline{\lambda x \colon A.s} \:&=\: \lambda x\colon\overline{A}.\:\overline{s} &
    \overline{s\: t} \:&=\: \overline{s}\: \overline{t} \\
    \overline{\bot} \:&=\: \bot &
    \overline{\lnot s} \:&=\: \lnot \overline{s} \\
    \overline{s \limp t} \:&=\: \overline{s} \limp \overline{t} &
    \overline{s =_{A} t} \:&=\: A^*\: \overline{s}\: \overline{t} \\
    \overline{\forall x\colon A.s} \:&=\: \forall x\colon\overline{A}.\: A^*\: x\: x \limp \overline{s} &
    \overline{x\colon A} \:&=\: x\colon \overline{A}, A^*\: x\: x
  \end{align*}
  \begin{align*}
    \overline{a\colon \dpi x_1\colon A_1.\:\cdots\: \dpi x_n\colon A_n.\:\tp} \:&=\: a\colon\tp, a^*\colon \overline{A_1} \to \cdots \to \overline{A_n} \to a \to a \to o, a_{\per} \\
    o^*\: s\: t \:&=\: s =_o t \\
    (a\: t_1 \:\dots\: t_n)^*\: s\: t \:&=\: a^*\: \overline{t_1}\: \dots \: \overline{t_n}\: s\: t \\
    (\dpi x\colon A.B)^*\: s\: t \:&=\: \forall x,y\colon\overline{A}.\: A^*\: x\: y \limp B^*\: (s\: x)\: (t\: y)
  \end{align*}
  Here, $a_{\per}$ is defined as follows:
  \[a_{\per} = \forall  x_1\colon\overline{A_1}.\: \dots \forall x_n\colon\overline{A_n}.\: \forall u,v\colon a.\: a^*\:x_1\:\dots\:x_n\:u\:v \limp u =_a v\]
\end{definition}

\begin{theorem}[Completeness \cite{RRB23}]
  \label{thm:erasure_completeness}
  \begin{itemize}
  \item if $\Gamma \vdash^{\mathsf{d}}_T A\colon\tp$ then
    $\overline{\Gamma} \vdash^{\mathsf{s}}_{\overline{T}}
    \overline{A}\colon\tp$ and $A^*$ is a PER over $\overline{A}$
  \item if $\Gamma \vdash^{\mathsf{d}}_T A \equiv B$ then
    $\overline{\Gamma} \vdash^{\mathsf{s}}_{\overline{T}} \forall
    x,y\colon\overline{A}.\: A^*\: x\: y =_o B^*\: x\: y$
  \item if $\Gamma \vdash^{\mathsf{d}}_T s\colon A$ then
    $\overline{\Gamma} \vdash^{\mathsf{s}}_{\overline{T}}
    \overline{s}\colon\overline{A}$ and
    $\overline{\Gamma} \vdash^{\mathsf{s}}_{\overline{T}} A^*\:
    \overline{s}\: \overline{s}$
  \item if $\Gamma \vdash^{\mathsf{d}}_T s$ then
    $\overline{\Gamma} \vdash^{\mathsf{s}}_{\overline{T}}
    \overline{s}$
  \end{itemize}
\end{theorem}

\begin{theorem}[Soundness \cite{RRB23}]
  \label{thm:erasure_soundness}
  \begin{itemize}
  \item if $\Gamma \vdash^{\mathsf{d}}_T s\colon o$ and
    $\overline{\Gamma} \vdash^{\mathsf{s}}_{\overline{T}}
    \overline{s}$ then $\Gamma \vdash^{\mathsf{d}}_T s$
  \item if $\Gamma \vdash^{\mathsf{d}}_T s\colon A$ and
    $\Gamma \vdash^{\mathsf{d}}_T t\colon A$ and
    $\overline{\Gamma} \vdash^{\mathsf{s}}_{\overline{T}} A^*\:
    \overline{s}\: \overline{t}$ then
    $\Gamma \vdash^{\mathsf{d}}_T s =_{A} t$
  \end{itemize}
\end{theorem}

Note that the erasure treats simple types and dependent types in the
same way.  In the following, we define a post-processing function
$\Phi$ on top of the original erasure \cite{RRB23} which allows us to
erase to simpler but equivalent formulas. The goal of $\Phi$ is to
replace $A^*\: s\: t$ where $A$ is a simple type by $s =_{A} t$. As a
consequence, the guard $A^*\: x\: x$ in
$\overline{\forall x\colon A.s}$ for simple types $A$ can be
removed. The following definition gives a presentation of $\Phi$ as a
pattern rewrite system \cite{MN98}.

\begin{definition}
  \label{def:phi}
  Given a HOL term $s$, we define $\Phi(s)$ to be the HOL term which
  results from applying the following pattern rewrite rules
  exhaustively to all subterms in a bottom-up fashion:
  \begin{align*}
    a^*\: F\: G \quad&\to\quad F =_{a} G \\
    \forall x,y\colon A.\: (x =_{A} y) \limp (F\: x =_{B} G\:y) \quad&\to\quad F =_{A \to B} G \\
    \forall x\colon A.\: (x =_{A} x) \limp F\: x \quad&\to\quad \forall x\colon A.\: F\: x
  \end{align*}
  Here, $F,G$ are free variables for terms, $a^*$ denotes the constant
  for the PER of a simple base type $a$ and $A,B$ are placeholders for
  simple types. Given a HOL theory $T$, there are finitely many
  instances for $a^*$ but infinite choices for $A$ and $B$, so the
  pattern rewrite system is infinite.
\end{definition}

\begin{lemma}
  \label{lem:phi_correct}
  Assume
  $\Gamma \vdash^{\mathsf{d}}_T s\colon o$.
  $\Gamma \vdash^{\mathsf{d}}_T s$ if and only if
  $\overline{\Gamma} \vdash^{\mathsf{s}}_{\overline{T}}
  \Phi(\overline{s})$.
\end{lemma}

\begin{proof}
  Since the erasure is sound and complete (\thmref{erasure_soundness}
  and \thmref{erasure_completeness}), it suffices to show that
  $\overline{\Gamma} \vdash^{\mathsf{s}}_{\overline{T}}
  \Phi(\overline{s})$ if and only if
  $\overline{\Gamma} \vdash^{\mathsf{s}}_{\overline{T}} \overline{s}$.
  Consider the rules from \defref{phi}. $\Phi(\overline{s})$ is
  well-defined: Clearly, the rules terminate and confluence follows
  from the lack of critical pairs \cite{MN98}.  Hence,
  it is sufficient to prove
  $\smash{\overline{\Gamma} \vdash^{\mathsf{s}}_{\overline{T}} l}$ if and only
  if $\smash{\overline{\Gamma} \vdash^{\mathsf{s}}_{\overline{T}} r}$ for
  every rule in \defref{phi}.  For the first rule, assume
  $\smash{\overline{\Gamma} \vdash^{\mathsf{s}}_{\overline{T}} a^*\: F\: G}$.
  Since $\smash{a_{\per} \in \overline{T}}$, we have
  $\smash{\overline{\Gamma} \vdash^{\mathsf{s}}_{\overline{T}} F =_a G}$.  Now
  assume
  $\smash{\overline{\Gamma} \vdash^{\mathsf{s}}_{\overline{T}} F =_a G}$.
  Since $F$ has type $a$, we obtain
  $\Gamma \vdash^{\mathsf{d}}_T F =_a F$. Completeness of the erasure
  yields $\smash{\overline{\Gamma} \vdash^{\mathsf{s}}_{\overline{T}} a^*\: F\:F}$.
  Now, the assumption allows us to
  replace equals by equals, so we conclude
  $\overline{\Gamma} \vdash^{\mathsf{s}}_{\overline{T}} a^*\: F\: G$.
  The desired result for the second rule follows from extensionality.
  Finally, the third rule is an easy logical simplification. \qed
\end{proof}

Given a theory $T$ (context $\Gamma$) we write
$\Phi(\overline{T})$ ($\Phi(\overline{\Gamma})$) to denote its erased
version where formulas have been simplified with $\Phi$.

\begin{corollary}
  \label{cor:phi_correct}
    Assume
    $\Gamma \vdash^{\mathsf{d}}_T s\colon o$.
    $\Gamma \vdash^{\mathsf{d}}_T s$ if and only if
    $\Phi(\overline{\Gamma}) \vdash^{\mathsf{s}}_{\Phi(\overline{T})}
    \Phi(\overline{s})$.
\end{corollary}

\begin{example}
  Consider again the axiom recursively defining $\app$ from
  \exaref{plus_app}
  \begin{align*}
    & \forall n,m\colon\nat,z\colon\elem,x\colon\lst\: n,y\colon\lst\: m. \\
    & \hspace{.6cm}\app\: (\suc\: n)\: m\: (\cons\: n\: z\: x)\: y =_{\lst\: (\suc\: (\plus\: n\: m))} \cons\: (\plus\: n\: m)\: z\: (\app\: n\: m\: x\: y)
  \end{align*}
  which we refer to as $s_{\app}$. Its post-processed erasure
  $\Phi(\overline{s_{\app}})$ is given by the following formula
  which is simpler than $\overline{s_{\app}}$:
  \begin{align*}
    & \forall n,m\colon\nat,z\colon\elem, x\colon\lst.\: \lst^*\: n\: x\: x \limp 
    \forall y\colon\lst.\: \lst^*\: m\: y\: y \limp {} \lst^*\: \bigl(\suc\: (\plus\: n\: m)\bigr)\\
    & \hspace{.6cm}\bigl(\app\: (\suc\: n)\: m\: (\cons\: n\: z\: x)\: y\bigr)\: \bigl(\cons\: (\plus\: n\: m)\: z\: (\app\: n\: m\: x\: y)\bigr)
  \end{align*}                                                         
\end{example}

\section{Tableau Calculus for DHOL}
\label{sec:calculus}

\subsection{Rules}

The tableau calculus from \cite{BS10,BB11} is the basis of Satallax
\cite{B12} and its fork Lash \cite{BK22}. We present an extension of
this calculus from HOL to DHOL by extending the rules to DHOL as well
as providing tableau rules for the translation from DHOL to HOL. A
\emph{branch} is a 3-tuple $(T,\Gamma,\Gamma')$ which is
\emph{well-formed} if $\vdash^{\mathsf{d}} T\:\thy$,
$\vdash^{\mathsf{d}}_T \Gamma\:\ctx$ and
$\vdash^{\mathsf{s}}_{\smash{\Phi(\overline{T})}} \Gamma'\:\ctx$.
Intuitively, the theory contains the original problem and remains
untouched while the contexts grow by the application of
rules. Furthermore, DHOL and HOL are represented separately: For DHOL,
the theory $T$ and context $\Gamma$ are used while HOL has a separate
context $\Gamma'$ with respect to the underlying theory
$\Phi(\overline{T})$. In particular, each rule in \figref{trules}
really stands for two rules: one that operates in DHOL and the
original version that operates in HOL. Except for the erasure rules
$\erar$ and $\erbr$ which add formulas to the HOL context based on
information from the DHOL theory and context, the rules always stay in
DHOL or HOL, respectively. More formally, a \emph{step} is an
$n+1$-tuple
$\langle
(T,\Gamma,\Gamma'),(T,\Gamma_1,\Gamma'_1),\dots,(T,\Gamma_n,\Gamma'_n)
\rangle$ of branches where $\bot \not\in T,\Gamma,\Gamma'$ and either
$\Gamma \subset \Gamma_i$ and $\Gamma' = \Gamma_i'$ for all
$1 \leq i \leq n$ or $\Gamma = \Gamma_i$ and
$\Gamma' \subset \Gamma_i'$ for all $1 \leq i \leq n$. Given a step
$\langle A, A_1,\dots,A_n \rangle$, the branch $A$ is called its
\emph{head} and each $A_i$ is an \emph{alternative}.

A \emph{rule} is a set of steps defined by a schema. For example, the
rule $\impr$ from \figref{trules} indicates the set of steps
$\langle
(T,\Gamma,\Gamma'),(T,\Gamma_1,\Gamma'_1),(T,\Gamma_2,\Gamma'_2)
\rangle$ where $\bot \not\in T,\Gamma,\Gamma'$ and either
$s \limp t \in T,\Gamma$ or
$s \limp t \in \Phi(\overline{T}),\Gamma'$.  In the former case, we
have $\Gamma_1 = \Gamma,\lnot s$ and $\Gamma_2 = \Gamma,t$ as well as
$\Gamma' = \Gamma'_1 = \Gamma'_2$.  The latter case is the same but
with the primed and unprimed variants swapped.

In the original tableau calculus \cite{BS10,BB11}, normalization is
defined with respect to an axiomatized generic operator $[\cdot]$. As
one would expect, one of these axioms states that the operator does
not change the semantics of a given term. Since there is no formal
definition of DHOL semantics yet, we simply use $[s]$ to denote the
$\beta\eta$-normal form of $s$ which is in accordance with our
implementation.

A rule \emph{applies} to a branch $A$ if some step in the rule has $A$
as its head. A tableau calculus is a set of steps. Let $\xT$ be the
tableau calculus defined by the rules in \figref{trules}.  The side
condition of freshness in $\nallr$ means that for a given step with
head $(T,\Gamma,\Gamma')$ there is no type $A$ such that
$y\colon A \in T,\Gamma$ or $y\colon A \in \Phi(\overline{T}),\Gamma'$
and we additionally require that there is no name $x$ such that
$\lnot [s\: x] \in T,\Gamma$ or
$\lnot [s\: x] \in \Phi(\overline{T}),\Gamma'$.  In practice, this
means that to every formula, $\nallr$ can be applied at most
once. Furthermore, the side condition $t\colon A$ in the rule $\allr$
means that either $\Gamma \vdash^{\mathsf{d}}_T t\colon A$ or
$\Gamma' \vdash^{\mathsf{s}}_{\smash{\Phi(\overline{T})}} t\colon A$ depending
on whether the premise is in $T,\Gamma$ or
$\Phi(\overline{T}),\Gamma'$. The side condition
$\overline{s}\colon o$ in the rule $\erar$ means that
$\Gamma' \vdash^{\mathsf{s}}_{\smash{\Phi(\overline{T})}} \overline{s}\colon
o$. This is to prevent application of $\erar$ before the necessary
type information is obtained by applying $\erbr$.

The set of \emph{$\xT$-refutable} branches is defined inductively: If
$\bot \in T,\Gamma,\Gamma'$, then $(T,\Gamma,\Gamma')$ is
refutable. If $\langle A,A_1,\dots,A_n \rangle$ is a step in $\xT$ and
every alternative $A_i$ is refutable, then $A$ is refutable.

\begin{figure}[t]
\addtolength{\jot}{1ex}
\begin{gather*}
         \negr\  \frac{s, \lnot s}{\bot}
  \qquad \neqr\  \frac{s \neq_{a t_1 \dots t_n} s}{\bot}
  \qquad \dnegr\ \frac{\lnot\lnot s}{s}
  \qquad \impr\  \frac{s \limp t}{\lnot s \mid t}
  \qquad \nimpr\ \frac{\lnot(s \limp t)}{s,\lnot t}
  \\     \allr\  \frac{\forall x\colon A.s}{[s[x/t]]}\
                 t\colon A
  \qquad \nallr\ \frac{\lnot\forall x\colon A.s}{y\colon A,\lnot[s[x/y]]}\
                 y \text{ fresh}
  \qquad \ber\   \frac{s \neq_{o} t}{s,\lnot t \mid \lnot s,t}
  \\     \bqr\   \frac{s =_{o} t}{s,t \mid \lnot s,\lnot t}
  \quad\: \fer\   \frac{s \neq_{\dpi x\colon A.B} t}{\lnot[\forall x.\:sx=tx]}\
                 x \not\in \xV s \cup \xV t
  \quad\: \fqr\   \frac{s =_{\dpi x\colon A.B} t}{[\forall x.\:sx=tx]}\
                 x \not\in \xV s \cup \xV t
  \\     \matr\  \frac{x\: s_1\:\dots\: s_n,\lnot\: x\: t_1\:\dots\: t_n}
                      {s_1 \neq t_1 \mid \cdots \mid s_n \neq t_n}\
                 n \geq 1
  \qquad \decr\  \frac{x\: s_1\:\dots\: s_n \neq_{a u_1 \dots u_n} x\: t_1 \:\dots\: t_n}
                      {s_1 \neq t_1 \mid \cdots \mid s_n \neq t_n}\
                 n \geq 1
  \\     \conr\  \frac{s =_{a t_1\dots t_n} t, u \neq_{a t_1 \dots t_n} v}
                      {s \neq u, t \neq u \mid s \neq v, t \neq v}
  \qquad \erar\  \frac{s}
                      {[\Phi(\overline{s})]}\
                 \overline{s}\colon o
  \qquad \erbr\  \frac{x\colon A}
                      {\ds x\colon \overline{A}, A^*\: x\: x}
\end{gather*}
\caption{Tableau rules for DHOL}
\label{fig:trules}
\end{figure}

The rules in \figref{trules} strongly resemble the tableau calculus
from \cite{BB11}. In order to support DHOL, we replaced simple types
by their dependent counterparts. To that end, we tried to remain as
simple as possible by only allowing syntactically equivalent types in
$\allr$ and $\conr$: Adding a statement like $A \equiv A'$ as a
premise would change the tableau calculus as well as the automated
proof search significantly, so these situations are handled by the
erasure for which the additional rules $\erar$, $\erbr$ are
responsible.

It is known that the restriction of $\xT$ to HOL (without $\erar$ and
$\erbr$) is sound and complete with respect to Henkin semantics
\cite{BS10,BB11}. Furthermore, due to \corref{phi_correct}, the
rules $\erar$ and $\erbr$ define a sound and complete translation from
DHOL to HOL with respect to Rothgang et al.'s definition of
provability in DHOL \cite{RRB23}.

\subsection{Soundness and Completeness}

In general, a soundness result based on the refutability of a branch
$(T,\Gamma,\Gamma')$ is desirable. If there were a definition of
semantics for DHOL which is a conservative extension of Henkin
semantics, the proof could just refer to satisfiability of
$T,\Gamma,\Gamma'$. Unfortunately, this is not the case.  Note that an
appropriate definition of semantics is out of the scope of this paper:
In addition to its conception, we would have to prove soundness and
completeness of $\vdash^{\mathsf{d}}$ on top of the corresponding
proofs for our novel tableau calculus. Therefore, soundness and
completeness of the tableau calculus will be established with respect
to provability in DHOL or HOL.  Unfortunately, this requirement
complicates the proof tremendously as a refutation can contain a
mixture of DHOL, erasure and HOL rules. Therefore, we have to consider
both HOL and DHOL and need to establish a correspondence between
$\Gamma$ and $\Gamma'$ which is difficult to put succinctly and seems
to be impossible without further restricting the notion of a
well-formed branch.  Therefore, we prove soundness and completeness
with respect to a notion of refutability which has three stages: At
the beginning, only DHOL rules are applied, the second stage is solely
for the erasure and in the last phase, only HOL rules are
applied. Note that this notion of refutability includes the sound but
incomplete strategy of only using native DHOL rules as well as the
sound and complete strategy of exclusively working with the erasure.

\begin{definition}
  A branch $(T,\Gamma,\Gamma')$ is \emph{s-refutable} if it is
  refutable with respect to the HOL rules.
\end{definition}

\begin{lemma}
  \label{lem:s_refut}
  A well-formed branch $(T,\Gamma,\Gamma')$ is s-refutable $\iff$
  $\Gamma' \vdash^{\mathsf{s}}_{\Phi(\overline{T})} \bot$.
\end{lemma}

\begin{proof}
  Immediate from soundness and completeness of the original HOL
  calculus as well as soundness and completeness of
  $\vdash^{\mathsf{s}}$. \qed
\end{proof}

\begin{definition}
  The set of \emph{e-refutable} branches is inductively defined as
  follows: If $(T,\Gamma,\Gamma')$ is s-refutable and
  $\Gamma' \subseteq \Phi(\overline{\Gamma})$, then it is e-refutable.
  If $\langle A,A_1 \rangle \in \erar \cup \erbr$ and $A_1$ is
  e-refutable, then $A$ is e-refutable.
\end{definition}

\begin{lemma}
  \label{lem:e_refut}
  If $(T,\Gamma,\Gamma')$ is well-formed and e-refutable then
  $\Phi(\overline{\Gamma}) \vdash^{\mathsf{s}}_{\Phi(\overline{T})}
  \bot$.
\end{lemma}

\begin{proof}
  Let $(T,\Gamma,\Gamma')$ be well-formed and e-refutable. We proceed
  by induction on the definition of e-refutability.  If
  $(T,\Gamma,\Gamma')$ is s-refutable then
  $\Gamma' \vdash^{\mathsf{s}}_{\smash{\Phi(\overline{T})}} \bot$ by
  \lemref{s_refut}.  Since $\Gamma' \subseteq \Phi(\overline{\Gamma})$
  we also have
  $\Phi(\overline{\Gamma}) \vdash^{\mathsf{s}}_{\smash{\Phi(\overline{T})}}
  \bot$.  For the induction step, let
  $\langle (T,\Gamma,\Gamma'),(T,\Gamma,\Gamma'_1)\rangle$ be a step
  with either $\erar$ or $\erbr$ and assume that the branch
  $(T,\Gamma,\Gamma'_1)$ is e-refutable. Since well-formedness of
  $(T,\Gamma,\Gamma'_1)$ follows from the well-formedness of
  $(T,\Gamma,\Gamma')$, the induction hypothesis yields
  $\Phi(\overline{\Gamma}) \vdash^{\mathsf{s}}_{\smash{\Phi(\overline{T})}}
  \bot$ as desired. \qed
\end{proof}

\begin{definition}
  The set of \emph{d-refutable} branches is inductively defined as
  follows: If $(T,\Gamma,\cdot)$ is e-refutable or
  $\bot \in T,\Gamma$, then it is d-refutable.  If
  $\langle A,A_1,\dots,A_n \rangle \in \xT \setminus (\erar \cup
  \erbr)$ and every alternative $A_i$ is d-refutable, then $A$ is
  d-refutable.
\end{definition}

Next, we have to prove soundness of every DHOL rule. For most of the
rules, this is rather straightforward. We show soundness of $\fer$,
$\fqr$ and $\decr$ as representative cases and start with an auxiliary
lemma.

\begin{lemma}
  \label{lem:bracket}
  Assume $\Gamma \vdash^{\mathsf{d}}_T s\colon o$. We have
  $\Gamma \vdash^{\mathsf{d}}_T s$ if and only if
  $\Gamma \vdash^{\mathsf{d}}_T [s]$.
\end{lemma}

\begin{proof}
  By the beta and eta rules, we have
  $\Gamma \vdash^{\mathsf{d}}_T s =_o [s]$.  Using cong$\vdash$ we
  obtain the desired result in both directions. \qed
\end{proof}

\begin{lemma}[$\fer$]
  \label{lem:soundness_fe}
  Let $(T,\Gamma,\Gamma')$ be a well-formed branch.  Choose $x$ such
  that $x \not\in \xV s \cup \xV t$ and assume
  $s \neq_{\dpi x\colon A.B} t \in T,\Gamma$.  If
  $\Gamma,\lnot [\forall x\colon A. sx = tx] \vdash^{\mathsf{d}}_T
  \bot$ then $\Gamma \vdash^{\mathsf{d}}_T \bot$.
\end{lemma}

\begin{proof}
  From the assumptions and \lemref{bracket}, we obtain
  $\Gamma \vdash^{\mathsf{d}}_T s \neq_{\dpi x\colon A.B} t$ and
  $\Gamma \vdash^{\mathsf{d}}_T \forall x\colon A. sx =_B tx$.
  Furthermore, an application of $\forall$e yields
  $\Gamma, x\colon A \vdash^{\mathsf{d}}_T sx =_B tx$.  Using
  cong$\lambda$, we get
  $\Gamma \vdash^{\mathsf{d}}_T (\lambda x\colon A.sx) =_{\dpi x\colon
    A.B} (\lambda x\colon A.tx)$.  Hence, we can apply eta
  ($x \not\in \xV s \cup \xV t$) , sym and the admissible rule trans
  \cite{RRB23ext} which says that equality is transitive to get
  $\Gamma \vdash^{\mathsf{d}}_T s =_{\dpi x\colon A.B} t$ and
  therefore $\Gamma \vdash^{\mathsf{d}}_T \bot$. \qed
\end{proof}

\begin{lemma}[$\fqr$]
  \label{lem:soundness_fq}
  Let $(T,\Gamma,\Gamma')$ be a well-formed branch.  Assume
  $s =_{\dpi x\colon A.B} t \in T,\Gamma$ and
  $x \not\in \xV s \cup \xV t$.  If
  $\Gamma,[\forall x\colon A. sx = tx] \vdash^{\mathsf{d}}_T \bot$
  then $\Gamma \vdash^{\mathsf{d}}_T \bot$.
\end{lemma}

\begin{proof}
  From the assumptions,
  $\Gamma \vdash^{\mathsf{d}}_T \lnot[s] =_o [\lnot s]$, cong$\vdash$
  and \lemref{bracket}, we obtain
  $\Gamma \vdash^{\mathsf{d}}_T s =_{\dpi x\colon A.B} t$ and
  $\Gamma \vdash^{\mathsf{d}}_T \lnot \forall x\colon A. sx =_B tx$.
  Furthermore, we have
  $\Gamma, x\colon A \vdash^{\mathsf{d}}_T sx =_B tx$ by refl and
  congAppl.  Hence, $\forall$i yields
  $\Gamma \vdash^{\mathsf{d}}_T \forall x\colon A. sx =_B tx$ and we
  conclude by an application of $\lnot$e. \qed
\end{proof}

\begin{lemma}[$\decr$]
  \label{lem:soundness_dec}
  Let $(T,\Gamma,\Gamma')$ be a well-formed branch.
  Assume
  \[x\: s_1 \:\dots\: s_n \neq_{a u_1\dots u_m} x\: t_1 \:\dots\: t_n
    \in T,\Gamma\] and
  $\Gamma \vdash^{\mathsf{d}}_T x\colon \dpi y_1\colon A_1 \cdots \dpi
  y_n \colon A_n.a\: u_1' \:\dots\: u_m'$ where
  $u_i = u_i'[y_1/s_1 \dots y_n/s_n]$ for $1 \leq i \leq m$.  If
  $\Gamma,s_i \neq_{A_i[x_1/s_1,\dots,x_{i-1}/s_{i-1}]} t_i
  \vdash^{\mathsf{d}}_T \bot$ for all $1 \leq i \leq n$ then
  $\Gamma \vdash^{\mathsf{d}}_T \bot$.
\end{lemma}

\begin{proof}
  From the assumptions, we obtain
  $\Gamma \vdash^{\mathsf{d}}_T s_i
  =_{A_i[x_1/s_1,\dots,x_{i-1}/s_{i-1}]} t_i$ for all
  $1 \leq i \leq n$ and
  $\Gamma \vdash^{\mathsf{d}}_T x =_{\dpi y_1\colon A_1 \cdots \dpi
    y_n \colon A_n.a u_1' \dots u_m'} x$.  Hence, $n$ applications of
  the congruence rule for application yield
  $\Gamma \vdash^{\mathsf{d}}_T x\: s_1 \:\dots\: s_n =_{a u_1 \dots
    u_m} x\: t_1 \:\dots\: t_n$. Since we also have
  $\Gamma \vdash^{\mathsf{d}}_T x\: s_1 \:\dots\: s_n \neq_{a u_1\dots
    u_m} x\: t_1 \:\dots\: t_n$, we obtain
  $\Gamma \vdash^{\mathsf{d}}_T \bot$. \qed
\end{proof}

Now we are ready to prove the soundness result for $\xT$.

\begin{theorem}
  \label{thm:tableau_soundness}
  If $(T,\Gamma,\cdot)$ is well-formed and d-refutable then
  $\Gamma \vdash^{\mathsf{d}}_T \bot$.
\end{theorem}

\begin{proof}
  Let $(T,\Gamma,\cdot)$ be well-formed and d-refutable. We proceed by
  induction on the definition of d-refutability.  If
  $(T,\Gamma,\cdot)$ is e-refutable, the result follows from
  \lemref{e_refut} together with \corref{phi_correct}.  If
  $\bot \in T,\Gamma$ then clearly
  $\Gamma \vdash^{\mathsf{d}}_T \bot$.  For the inductive case,
  consider a step
  $\langle
  (T,\Gamma,\cdot),(T,\Gamma_1,\cdot),\dots,(T,\Gamma_n,\cdot)
  \rangle$ with some DHOL rule.  Since $(T,\Gamma,\cdot)$ is
  d-refutable, all alternatives must be d-refutable.  If we manage to
  show well-formedness of every alternative, we can apply the
  induction hypothesis to obtain $\Gamma_i \vdash^{\mathsf{d}}_T \bot$
  for all $1 \leq i \leq n$.  Then, we can conclude
  $\Gamma \vdash^{\mathsf{d}}_T \bot$ by soundness of the DHOL rules.
  Hence, it remains to prove well-formedness of the alternatives. In
  most cases, this is straightforward. We only show one interesting
  case, namely $\decr$.

  Instead of proving
  $\Gamma,s_i \neq_{A_i[x_1/s_1,\dots,x_{i-1}/s_{i-1}]} t_i
  \vdash^{\mathsf{d}}_T \bot$ for all $1 \leq i \leq n$ we show that
  $\Gamma \vdash^{\mathsf{d}}_T s_i
  =_{A_i[x_1/s_1,\dots,x_{i-1}/s_{i-1}]} t_i$ for all
  $1 \leq i \leq n$.  Since $(T,\Gamma,\cdot)$ is a well-formed
  branch, both $s_1$ and $t_1$ have type $A_1$. Hence,
  $(T,(\Gamma,s_1 \neq_{A_1} t_1),\cdot)$ is well-formed and our
  original induction hypothesis yields
  $\Gamma, s_1 \neq_{A_1} t_1 \vdash^{\mathsf{d}}_T \bot$ from which
  we obtain $\Gamma \vdash^{\mathsf{d}}_T s_1 =_{A_1} t_1$.  Now let
  $i \leq n$ and assume we have
  $\Gamma \vdash^{\mathsf{d}}_T s_j
  =_{A_j[x_1/s_1,\dots,x_{j-1}/s_{j-1}]} t_j$ for all $j < i$
  ($\ast$).  This is only possible if
  $\Gamma \vdash^{\mathsf{d}}_T t_j\colon
  A_j[x_1/s_1,\dots,x_{j-1}/s_{j-1}]$ for all $j < i$. Since
  $(T,\Gamma,\cdot)$ is a well-formed branch, it is clear that
  $\Gamma \vdash^{\mathsf{d}}_T s_{i}\colon
  A_{i}[x_1/s_1,\dots,x_{i-1}/s_{i-1}]$ and
  $\Gamma \vdash^{\mathsf{d}}_T t_{i}\colon
  A_{i}[x_1/t_1,\dots,x_{i-1}/t_{i-1}]$.  From ($\ast$), we obtain
  \[\Gamma \vdash^{\mathsf{d}}_T t_{i}\colon A_{i}[x_1/s_1,\dots,x_{i-1}/s_{i-1}],\]
  so
  $(T,(\Gamma,s_i \neq_{A_i[x_1/s_1,\dots,x_{i-1}/s_{i-1}]}
  t_i),\cdot)$ is well-formed.  Hence, the original induction
  hypothesis yields
  $\Gamma \vdash^{\mathsf{d}}_T s_i =_{A_i[x_1/s_1,\dots,x_i/s_{i-1}]}
  t_i$ as desired. \qed
\end{proof}

In the previous proof, we can see that for $\decr$, well-formedness of
an alternative depends on refutability of all branches to the
left. Note that the same holds for $\matr$ and $\impr$. This is a
distinguishing feature of DHOL as in tableaux, branches are usually
considered to be independent.

Finally, completeness is immediate from the completeness of the HOL
tableau calculus and the erasure:

\begin{theorem}
  \label{thm:tableau_completeness} If
  $\Gamma \vdash^{\mathsf{d}}_T \bot$ then $(T,\Gamma,\cdot)$ is
  d-refutable.
\end{theorem}

\begin{proof}
  Let $\Gamma \vdash^{\mathsf{d}}_T \bot$.  Using \corref{phi_correct}
  and \lemref{s_refut} we conclude s-refutability of
  $(T,\Gamma,\Phi(\overline{\Gamma}))$. By definition,
  $(T,\Gamma,\Phi(\overline{\Gamma}))$ is also e-refutable.
  Furthermore, by inspecting $\erar$ and $\erbr$ we conclude that
  $(T,\Gamma,\cdot)$ is also e-refutable and therefore
  d-refutable. \qed
\end{proof}

\section{Implementation}
\label{sec:implementation}

We implemented the tableau calculus for DHOL as an extension of Lash
\cite{BK22} which is a fork of Satallax, a successful automated
theorem prover for HOL \cite{B12}. By providing an efficient C
implementation of terms with perfect sharing as well as other
important data structures and operations, Lash outperforms Satallax
when it comes to the basic ground tableau calculus which both of them
implement. However, Lash removes a lot of the additional features
beyond the basic calculus that was implemented in
Satallax. Nevertheless, this was actually beneficial for our purpose
as we could concentrate on adapting the core part. Note that Lash and
Satallax do not just implement the underlying ground tableau calculus
but make heavy use of SAT-solving and a highly customizable priority
queue to guide the proof search \cite{B13,B12}.

For the extension of Lash to DHOL, the data structure for terms had to
be changed to support dependent function types as well as quantifiers
and lambda abstractions with dependent types.  Of course, it would be
possible to represent everything in the language of DHOL but the
formulation of DHOL suggests that the prover should do as much as
possible in the HOL fragment and only use ``proper'' DHOL when it is
really necessary. With this in mind, the parser first always tries to
produce simply-typed terms and only resorts to dependent types when it
is unavoidable. Therefore, the input problem often looks like a
mixture of HOL and DHOL even though everything is included in DHOL. A
nice side effect of this design decision is that our extension of Lash
works exactly like the original version on the HOL fragment except for
the fact that it is expected to be slower due to the numerous case
distinctions between simple types and dependent types which are needed
in this setting.

Although DHOL is not officially part of TPTP THF, it can be expressed
due to the existence of the \texttt{!>}-symbol which is used for
polymorphism. Hence, a type $\dpi x\colon A. B$ is represented as
\texttt{!>[X:A]:B}.  For simplicity and efficiency reasons, we did not
implement dependent types by distinguishing base types from their term
arguments but represent the whole dependent type as a term. When
parsing a base type $a$, Lash automatically creates an eponymous
constant of type $\tp$ to be used in dependent types as well as a
simple base type $a_0$ for the erasure and a constant $a^*$ for its
PER. The flags \texttt{DHOL\_RULES\_ONLY} and
\texttt{DHOL\_ERASURE\_ONLY} control the availability of the erasure
as well as the native DHOL rules, respectively. Note that the
implementation is not restricted to d-refutability but allows for
arbitrary refutations. In the standard flag setting, however, only the
native DHOL rules are used. Clearly, this constitutes a sound
strategy. It is incomplete since the confrontation rule only considers
equations with syntactically equivalent types. We have more to say
about this in \secref{rule_impl}.

\subsection{Type Checking}

By default, problems are only type-checked with respect to their
simply-typed skeleton. If the option \texttt{exactdholtypecheck} is
set, type constraints stemming from the term arguments of dependent
base types are generated and added to the conjecture. The option
\texttt{typecheck\-only} discards the original conjecture, so Lash
just tries to prove the type constraints.  Since performing the type
check involves proper theorem proving, we added the new SZS ontology
statuses \texttt{TypeCheck} and \texttt{InexactTypecheck} to the
standardized output of Lash. Here, the former one means that a problem
type checks while the latter one just states that it type checks with
respect to the simply-typed skeleton.

For the generation of type constraints, each formula of the problem is
traversed like in normal type checking.  In addition, every time a
type condition $a\: t_1 \:\dots\: t_n \equiv a\: s_1 \:\dots\: s_n$
comes up and there is some $i$ such that $s_i$ and $t_i$ are not
syntactically equivalent, a constraint stating that $s_i = t_i$ is
provable is added to the set of type constraints.  Note that it does
not always suffice to just add $s_i = t_i$ as this equation may
contain bound variables or only hold in the context in which the
constraint appears.  To that end, we keep relevant information about
the term context when generating these constraints. Whenever a forall
quantifier or lambda abstraction comes up, it is translated to a
corresponding forall quantifier in the context since we want the
constraint to hold in any case. While details like applications can be
ignored, it is important to keep left-hand sides of implications in
the context as it may be crucial for the constraint to be met. In
general, any axiom may contribute to the typechecking proof.

\begin{example}
  The conjecture
  \[\forall n\colon\nat,x\colon\lst\:n.\: n =_{\nat} 0 \limp \app\: n\: n\: x\: x = x\]
  is well-typed if the type constraint
  \[\forall n\colon\nat,x\colon\lst\: n.\: n =_{\nat} 0 \limp \plus\:n\: n =_{\nat} n\]
  is provable. Lash can generate this constraint and finds a proof
  quickly using the axiom
  $\forall n\colon\nat.\:\plus\: \zro\: n =_{\nat} n$.
\end{example}

Since conjunctions and disjunctions are internally translated to
implications, it is important to note that we process formulas from
left to right, i.e.~for $x\colon \lst\: n$ and $y\colon \lst\: m$, the
proposition $m \neq n \lor x = y$ type checks because we can assume
$m = n$ to process $x = y$.  Consequently, $x = y \lor m \neq n$ does
not type check. As formulas are usually read from left to right, this
is a natural adaption of short-circuit evaluation in programming
languages. Furthermore, it is in accordance with the presentation of
Rothgang et al.~\cite{RRB23} as well as the corresponding
implementation in PVS \cite{ROS98}. As a matter of fact, PVS handles
its undecidable type checking problem in essentially the same way as
our new version of Lash by generating so called \emph{type correctness
  conditions} (TCCs).

\subsection{Implementation of the Rules}
\label{sec:rule_impl}

Given the appropriate infrastructure for dependent types, the
implementation of most rules in \figref{trules} is a straightforward
extension of the original HOL implementation.  For $\allr$, the side
condition $\Gamma \vdash^{\mathsf{d}}_T t\colon A$ is undecidable in
general.  It has been chosen to provide a simple characterization
of the tableau calculus. Furthermore, it emphasizes that we do not
instantiate with terms whose type does not literally match with the
type of the quantified variable. In the implementation, we keep a pool
of possible instantiations for types $A$ which occur in the
problem. The pool gets populated by terms of which we know that they
have a given type because this information was available during
parsing or proof search. Hence, we only instantiate with terms
$t$ for which we already know that
$\Gamma \vdash^{\mathsf{d}}_T t\colon A$ holds.

Given an equation $s =_A t$, there are many candidate representations
of $A$ modulo type equality.  When we build an equation in the
implementation, we usually use the type of the left-hand side. Since
all native DHOL rules of the tableau calculus enforce syntactically
equivalent types, the ambiguity with respect to the type of an
equation leads to problems. For example, consider a situation where
$\Gamma \vdash^{\mathsf{d}}_T s\colon A$,
$\Gamma \vdash^{\mathsf{d}}_T t\colon B$ and
$\Gamma \vdash^{\mathsf{d}}_T s =_A t$ which implies
$\Gamma \vdash^{\mathsf{d}}_T A \equiv B$. During proof search, it
could be that $\Gamma \vdash^{\mathsf{d}}_T t \neq s$ is
established. Clearly, this is a contradiction which leads to a
refutation, but usually the inequality annotated with
the type $B$ which makes the refutation inaccessible for our native
DHOL rules. Therefore, we implemented rules along the lines of
\[\symcastar\   \frac{s =_{A} t}{t =_{B} s}\ t\colon B \qquad \symcastbr\   \frac{s \neq_{A} t}{t \neq_{B} s}\ t\colon B\]
which do not only apply symmetry but also change the type of the
equality in a sound way.  Like in $\allr$, the side condition should
be read as $\Gamma \vdash^{\mathsf{d}}_T t\colon B$ which makes it
undecidable. However, in practice, we can compute a representative of
the type of $t$ given the available type information. While
experimenting with the new DHOL version of Lash, the implementation of
these rules proved to be very beneficial for refutations which only
work with the DHOL rules. For the future, it is important to note
that $\symcastar$ and $\symcastbr$ are not sound for the extension of
DHOL to predicate subtypes as $\Gamma \vdash^{\mathsf{d}}_T s =_{A} t$
and $\Gamma \vdash^{\mathsf{d}}_T t\colon B$ do not imply
$\Gamma \vdash^{\mathsf{d}}_T A \equiv B$ anymore.

\subsection{Generating Instantiations}

Since Lash implements a ground tableau calculus, it does not support
higher-order unification.  Therefore, the generation of suitable
instantiations is a major issue. In the case of DHOL, it is actually
beneficial that Lash already implements other means of generating
instantiations since the availability of unification for DHOL is
questionable: There exist unification procedures for dependent
type theories (see for example \cite{E89}) but for DHOL such a
procedure would also have to address the undecidable type equality
problem.

For simple base types, it suffices to consider so-called
\emph{discriminating} terms to remain complete \cite{BB11}. A term $s$
of simple base type $a$ is discriminating in a branch $A$ if
$s \neq_{a} t \in A$ or $t \neq_{a} s \in A$ for some term $t$. For
function terms, completeness is guaranteed by enumerating all possible
terms of a given type. Of course, this is highly impractical, and
there is the important flag
\texttt{INITIAL\_SUBTERMS\_AS\_INSTANTIATIONS} which adds all subterms
of the initial problem as instantiations. This heuristic works very
well in many cases.

For dependent types, we do not check for type equality when
instantiating quantifiers but only use instantiations with the exact
same type (c.f.~$\allr$ in \figref{trules}) and let the erasure handle
the remaining cases.

An interesting feature of this new version of Lash is the possibility
to automatically generate instantiations for induction axioms. Given
the constraints of the original implementation, the easiest way to
sneak a term into the pool of instantiations is to include it into an
easily provable lemma and then use the flag
\texttt{INITIAL\_SUBTERMS\_AS\_INSTANTIATIONS}. However, this adds
unnecessary proof obligations, so we modified the implementation such
that initial subterms as instantiations also include
lambda-abstractions corresponding to forall quantifiers.

\begin{example}
  \label{exa:list_induct}
  Consider the induction axiom for lists:
  \begin{align*}
    &\forall p\colon (\dpi n\colon\nat.\:\lst\:n \to o).\: p\:\zro\:\nil \\
    &\phantom{\forall p}\limp{} (\forall n\colon\nat,x\colon\elem,y\colon\lst\: n.\: p\: n\: y \limp p\:(\suc\: n)\: (\cons\: n\: x\: y)) \\
    &\phantom{\forall p}\limp{} (\forall n\colon\nat, x\colon\lst\: n.\: p\: n\: x)
  \end{align*}
  Even though it works for arbitrary predicates $p$, it is very hard
  for an ATP system to guess the correct instance for a given problem
  without unification in general. However, given the conjecture
  $\forall n\colon\nat,x\colon\lst\: n.\: \app\: n\: \zro\: x\: \nil =_{\lst\: n} x$
  we can easily read off the correct instantiation for $p$ where
  $\forall$ is replaced by $\lambda$.
\end{example}

\section{Case Study: List Reversal Is an Involution}
\label{sec:experiments}

Consider the following equational definition of the list reversal
function $\rev$:
\begin{align*}
  & \rev\:\zro\:\nil =_{\lst\:\zro} \nil \\
  & \forall n\colon\nat, x\colon\elem, y\colon \lst\: n. \\
  & \hspace{.6cm}\rev\: (\suc\: n)\: (\cons\: n\: x\: y) =_{\lst\: (\suc n)} \app\: n\: (\suc\: \zro)\: (\rev\: n\: y)\: (\cons\:\zro\: x\:\nil)
\end{align*}
The conjecture
\[\forall n\colon\nat, x\colon \lst\: n.\: \rev\: n\: (\rev\: n\: x) =_{\lst\: n} x \label{eq:rev_invol} \tag{\texttt{rev-invol}}\]
is very easy to state, but turns out to be hard to prove
automatically. The proof is based on the equational definitions of
$\plus$ and $\app$ given in \exaref{plus_app} as well as several
induction proofs on lists using the axiom from \exaref{list_induct}.
In particular, some intermediate goals are needed to succeed:
\begin{align*}
  & \forall n\colon\nat,x\colon \lst\: n.\: \app\: n\:\zro\: x\:\nil =_{\lst\: n} x \label{eq:app_nil} \tag{\texttt{app-nil}} \\[3mm]
  & \forall n_1\colon\nat, x_1\colon \lst\: n_1, n_2\colon\nat, x_2\colon \lst\: n_2, n_3\colon\nat, x_3\colon \lst\: n_3. \label{eq:app_assoc} \tag{\texttt{app-assoc}} \\
  & \hspace{.6cm}\phantom{{}={}} \app\: n_1\: (\plus\: n_2\: n_3)\: x_1\: (\app\: n_2\: n_3\: x_2\: x_3) \\
  & \hspace{.6cm}{}={} \app\: (\plus\: n_1\: n_2)\: n_3\: (\app\: n_1\: n_2\: x_1\: x_2)\: x_3  \\[3mm]
  & \forall n\colon\nat, x\colon \lst\: n, y\colon\elem, m\colon\nat, z\colon \lst\: m. & \label{eq:app_assoc_m1} \tag{\texttt{app-assoc-m1}} \\
  & \hspace{3mm}\app\: (\plus\: n\: (\suc\: \zro))\: m\: (\app\: n\: (\suc\: \zro)\: x\: (\cons\:\zro\: y\:\nil))\: z = \app\: n\: (\suc\: m)\: x\: (\cons\: m\: y\: z) \\[3mm]
  & \forall n\colon\nat, x\colon \lst\: n, m\colon\nat, y\colon \lst\: m. \label{eq:rev_invol_lem} \tag{\texttt{rev-invol-lem}} \\
  & \hspace{.6cm}\rev\: (\plus\: n\: m)\: (\app\: n\: m\: (\rev\: n\: x)\: y) = \app\: m\: n\: (\rev\: m\: y)\: x
\end{align*}
Note that for
polymorphic lists, this is a standard example of an induction proof
with lemmas (see e.g.~\cite[Section 2.2]{NPW02}). In the
dependently-typed case, however, many intermediate equations would be
ill-typed in interactive theorem provers like Coq or Lean. In order to
succeed in automatically proving these problems, we had to break them
down into separate problems for the instantiation of the induction
axiom, the base case and the step case of the induction proofs. Often,
we further needed to organize these subproblems in manageable
steps. Overall, we created 34 TPTP problem files which are distributed
over the intermediate goals as shown in \tabref{prob_dist}.
\begin{table}[t]
  \centering
  \begin{tabular}{cc}
    \toprule
    Goal & Number of Problem Files \\
    \midrule
    \ref{eq:app_nil} & 4 \\
    \ref{eq:app_assoc} & 8 \\
    \ref{eq:app_assoc_m1} & 5 \\
    \ref{eq:rev_invol_lem} & 12 \\
    \ref{eq:rev_invol} & 5 \\
    \bottomrule
  \end{tabular}
  \label{tab:prob_dist}
  \caption{Amount of problem files per (intermediate) goal}
\end{table}
Note that already type checking these intermediate problems is not
trivial: All type constraints are arithmetic equations, and given the
Peano axioms, many of them need to be proven by induction
themselves. Since we are mainly interested in the dependently-typed
part, we added the needed arithmetical facts as axioms. Overall, the
problem files have up to 18 axioms including the Peano axioms,
selected arithmetical results, the defining equations of $\plus$,
$\app$ and $\rev$ as well as the list induction axiom. We left out
unnecessary axioms in many problem files to make the proof search
feasible.

With our new modes for DHOL which solely work with the native DHOL
rules, Lash can type check and prove all problems easily. If we turn
off the native DHOL rules and only work with the erasure using the
otherwise same modes with a 60s timeout, Lash can still typecheck all
problems but it only manages to prove 7 out of 34 problems. In order
to further evaluate the effectiveness of our new implementation, we
translated all problems from DHOL to HOL using the \emph{Logic
  Embedding
  Tool}\footnote{\url{https://github.com/leoprover/logic-embedding}},
which performs the erasure from~\cite{RRB23}. We then tested 16 other
HOL provers available on
\emph{SystemOnTPTP}\footnote{\url{https://tptp.org/cgi-bin/SystemOnTPTP}}
on the translated problems with a 60s timeout (without type
checking). We found that 5 of the 34 problems could only be solved by
the DHOL version of Lash, including one problem where it only needs 5
inference steps. Detailed results as well as means to reproduce them
are available on Lash's
website\footnote{\url{http://cl-informatik.uibk.ac.at/software/lash-dhol/}}
together with its source code.

\section{Conclusion}
\label{sec:conclusion}

Starting from the erasure from DHOL to HOL by Rothgang et
al.~\cite{RRB23}, we developed a sound and complete tableau calculus
for DHOL which we implemented in Lash. To the best of our knowledge,
this makes it the first standalone automated theorem prover for
DHOL. According to the experimental results, configurations where the
erasure is performed as a preprocessing step for a HOL theorem prover
can be outperformed by our new prover by solely using the native DHOL
rules. We hope that this development will raise further interest in
DHOL. Possible further work includes theoretical investigations such as the
incorporation of choice operators into the erasure as well as a
definition of the semantics of DHOL. Furthermore, it is desirable to
officially define the TPTP syntax for DHOL which then opens the
possibility of establishing a problem data set on which current and
future tools can be compared. Finally, we would like to extend Lash to
support predicate subtypes.  Rothgang et al.~already incorporated this
into the erasure but there is no corresponding syntactic support in
TPTP yet. In particular, this would get us much closer to powerful
automation support for systems like PVS.

\begin{credits}
  \subsubsection{\ackname}
  The results were supported by the Ministry of Education, Youth and
  Sports within the dedicated program ERC CZ under the project POSTMAN
  no.~LL1902.
  This work has also received funding from the European Union’s Horizon
  Europe research and innovation programme under grant agreement
  no.~101070254 CORESENSE as well as the ERC PoC grant no.~101156734
  \emph{FormalWeb3}. Views and opinions expressed are however those of
  the authors only and do not necessarily reflect those of the
  European Union or the Horizon Europe programme. Neither the European
  Union nor the granting authority can be held responsible for them.

  \subsubsection{\discintname}
  The authors have no competing interests to declare that are relevant
  to the content of this article.
\end{credits}

\bibliographystyle{splncs04}
\bibliography{references}

\begin{thebibliography}{10}
\providecommand{\url}[1]{\texttt{#1}}
\providecommand{\urlprefix}{URL }
\providecommand{\doi}[1]{https://doi.org/#1}

\bibitem{BB11}
Backes, J., Brown, C.E.: Analytic tableaux for higher-order logic with choice.
  Journal of Automated Reasoning  \textbf{47},  451--479 (2011).
  \doi{10.1007/s10817-011-9233-2}

\bibitem{BBTV23}
Bentkamp, A., Blanchette, J., Tourret, S., Vukmirović, P.: Superposition for
  higher-order logic. Journal of Automated Reasoning  \textbf{67}, ~10 (2023).
  \doi{10.1007/s10817-022-09649-9}

\bibitem{BDN09}
Bove, A., Dybjer, P., Norell, U.: A brief overview of {Agda} – a functional
  language with dependent types. In: Berghofer, S., Nipkow, T., Urban, C.,
  Wenzel, M. (eds.) Proc.\ 22nd International Conference on Theorem Proving in
  Higher Order Logics. LNCS, vol.~5674, pp. 73--78 (2009).
  \doi{10.1007/978-3-642-03359-9_6}

\bibitem{B12}
Brown, C.E.: Satallax: An automatic higher-order prover. In: Gramlich, B.,
  Miller, D., Sattler, U. (eds.) Proc.\ 6th International Joint Conference on
  Automated Reasoning. LNAI, vol.~7364, pp. 111--117 (2012).
  \doi{10.1007/978-3-642-31365-3_11}

\bibitem{B13}
Brown, C.E.: Reducing higher-order theorem proving to a sequence of {SAT}
  problems. Journal of Automated Reasoning  \textbf{51},  57--77 (2013).
  \doi{10.1007/s10817-013-9283-8}

\bibitem{BK22}
Brown, C.E., Kaliszyk, C.: Lash 1.0 (system description). In: Blanchette, J.,
  Kovács, L., Pattinson, D. (eds.) Proc.\ 11th International Joint Conference
  on Automated Reasoning. LNAI, vol. 13385, pp. 350--358 (2022).
  \doi{10.1007/978-3-031-10769-6_21}

\bibitem{BS10}
Brown, C.E., Smolka, G.: Analytic tableaux for simple type theory and its
  first-order fragment. Logical Methods in Computer Science  \textbf{6}(2),
  1--33 (2010). \doi{10.2168/LMCS-6(2:3)2010}

\bibitem{CCD17}
Castellan, S., Clairambault, P., Dybjer, P.: Undecidability of equality in the
  free locally cartesian closed category (extended version). Logical Methods in
  Computer Science  \textbf{13}(4),  1--38 (2017).
  \doi{10.23638/LMCS-13(4:22)2017}

\bibitem{C40}
Church, A.: A formulation of the simple theory of types. Journal of Symbolic
  Logic  \textbf{5}(2),  56--68 (1940). \doi{10.2307/2266170}

\bibitem{COQ}
{Coq Development Team}: The Coq Reference Manual, Release 8.18.0 (2023)

\bibitem{E89}
Elliott, C.M.: Higher-order unification with dependent function types. In:
  Dershowitz, N. (ed.) Proc.\ 3rd International Conference on Rewriting
  Techniques and Applications. LNCS, vol.~355, pp. 121--136 (1989).
  \doi{10.1007/3-540-51081-8_104}

\bibitem{ML84}
Martin-Löf, P.: Constructive mathematics and computer programming.
  Philosophical Transactions of the Royal Society of London. Series A,
  Mathematical and Physical Sciences  \textbf{312},  501--518 (1984).
  \doi{10.1098/rsta.1984.0073}

\bibitem{MN98}
Mayr, R., Nipkow, T.: Higher-order rewrite systems and their confluence.
  Theoretical Computer Science  \textbf{192}(1),  3--29 (1998).
  \doi{10.1016/S0304-3975(97)00143-6}

\bibitem{dMU21}
de~Moura, L., Ullrich, S.: The {Lean} 4 theorem prover and programming
  language. In: Platzer, A., Sutcliffe, G. (eds.) Proc.\ 28th International
  Conference on Automated Deduction. LNAI, vol. 12699, pp. 625--635 (2021).
  \doi{10.1007/978-3-030-79876-5_37}

\bibitem{NPW02}
Nipkow, T., Paulson, L.C., Wenzel, M.: Isabelle/HOL --- A Proof Assistant for
  Higher-Order Logic, LNCS, vol.~2283. Springer Berlin Heidelberg (2002).
  \doi{10.1007/3-540-45949-9}

\bibitem{PS99}
Pfenning, F., Schürmann, C.: System description: Twelf – a meta-logical
  framework for deductive systems. In: Ganzinger, H. (ed.) Proc.\ 16th
  International Conference on Automated Deduction. LNAI, vol.~1632, pp.
  202--206 (1999). \doi{10.1007/3-540-48660-7_14}

\bibitem{RRB23}
Rothgang, C., Rabe, F., Benzmüller, C.: Theorem proving in dependently-typed
  higher-order logic. In: Pientka, B., Tinelli, C. (eds.) Proc.\ 29th
  International Conference on Automated Deduction. LNAI, vol. 14132, pp.
  438--455 (2023). \doi{10.1007/978-3-031-38499-8_25}

\bibitem{RRB23ext}
Rothgang, C., Rabe, F., Benzmüller, C.: Theorem proving in dependently-typed
  higher-order logic -- extended preprint (2023).
  \doi{10.48550/arXiv.2305.15382}

\bibitem{ROS98}
Rushby, J., Owre, S., Shankar, N.: Subtypes for specifications: Predicate
  subtyping in {PVS}. IEEE Transactions on Software Engineering
  \textbf{24}(9),  709--720 (1998). \doi{10.1109/32.713327}

\bibitem{SB21}
Steen, A., Benzmüller, C.: Extensional higher-order paramodulation in
  {Leo-III}. Journal of Automated Reasoning  \textbf{65},  775--807 (2021).
  \doi{10.1007/s10817-021-09588-x}

\bibitem{S+16}
Swamy, N., Hri\c{t}cu, C., Keller, C., Rastogi, A., Delignat-Lavaud, A.,
  Forest, S., Bhargavan, K., Fournet, C., Strub, P.Y., Kohlweiss, M.,
  Zinzindohoue, J.K., Zanella-B\'{e}guelin, S.: Dependent types and
  multi-monadic effects in {F$^{\star}$}. In: Bodik, R., Majumdar, R. (eds.)
  Proc.\ 43rd Annual ACM SIGPLAN-SIGACT Symposium on Principles of Programming
  Languages. pp. 256--270 (2016). \doi{10.1145/2837614.2837655}

\end{thebibliography}

\appendix

\section{Soundness of DHOL Rules}
\label{appendixa}

In the following lemmas, we prove the remaining cases of soundness of
the DHOL tableau rules with respect to the natural deduction calculus
in \figref{dholndcalculus}. We begin with two auxiliary lemmas.

\begin{lemma}
  \label{lem:xmcases}
  Let $T$, $\Gamma$ and $s$ be such that
  $\vdash^{\mathsf{d}} T\:\thy$,
  $\vdash^{\mathsf{d}}_T \Gamma\:\ctx$
  and
  $\Gamma \vdash^{\mathsf{d}}_T s : o$.
  If $\Gamma, s\vdash^{\mathsf{d}}_T \bot$
  and $\Gamma, \neg s\vdash^{\mathsf{d}}_T \bot$,
  then $\Gamma \vdash^{\mathsf{d}}_T \bot$.
\end{lemma}

\begin{proof}
  Since $\Gamma, s\vdash^{\mathsf{d}}_T \bot$,
  we have $\Gamma \vdash^{\mathsf{d}}_T \neg s$.
  Since $\Gamma, \neg s\vdash^{\mathsf{d}}_T \bot$,
  we have $\Gamma \vdash^{\mathsf{d}}_T s$.
  Hence, $\Gamma \vdash^{\mathsf{d}}_T \bot$. \qed
\end{proof}

\begin{lemma}
  \label{lem:eqtype}
  If $\Gamma \vdash^{\mathsf{d}}_T (s =_A t)\colon o$ then
  $\Gamma \vdash^{\mathsf{d}}_T s\colon A$ and
  $\Gamma \vdash^{\mathsf{d}}_T t\colon A$.
\end{lemma}

\begin{proof}
  Let $\Gamma \vdash^{\mathsf{d}}_T (s =_A t)\colon o$.
  An inspection of the proof system in \figref{dholndcalculus}
  shows that this can only be inferred by the rule $=$type
  which requires $\Gamma \vdash^{\mathsf{d}}_T s\colon A$ and
  $\Gamma \vdash^{\mathsf{d}}_T t\colon A$. \qed
\end{proof}

\begin{lemma}[$\negr$]
  \label{lem:soundness_neg}
  Let $(T,\Gamma,\Gamma')$ be a well-formed branch.
  If $s \in T,\Gamma$
  and $\lnot s \in T,\Gamma$
  then $\Gamma \vdash^{\mathsf{d}}_T \bot$.
\end{lemma}

\begin{proof}
  Clearly, $\Gamma \vdash^{\mathsf{d}}_T s$ and
  $\Gamma \vdash^{\mathsf{d}}_T \lnot s$.  Using the $\lnot$e rule, we
  obtain $\Gamma \vdash^{\mathsf{d}}_T \bot$. \qed
\end{proof}

\begin{lemma}[$\neqr$]
  \label{lem:soundness_neq}
  Let $(T,\Gamma,\Gamma')$ be a well-formed branch.
  If $s \neq_{a t_1 \dots t_n} s \in T,\Gamma$
  then $\Gamma \vdash^{\mathsf{d}}_T \bot$.
\end{lemma}

\begin{proof}
  From $s \neq_{a t_1 \dots t_n} s \in T,\Gamma$ we obtain
  $\Gamma \vdash^{\mathsf{d}}_T (s \neq_{a t_1 \dots t_n} s)\colon o$.
  Note that this can only be the case if
  $\Gamma \vdash^{\mathsf{d}}_T (s =_{a t_1 \dots t_n} s)\colon o$.
  Then, \lemref{eqtype} yields
  $\Gamma \vdash^{\mathsf{d}}_T s\colon a\: t_1 \:\dots\: t_n$.
  Hence, we can use the refl rule to get
  $\Gamma \vdash^{\mathsf{d}} s =_{a t_1 \dots t_n} s$ and therefore
  $\Gamma \vdash^{\mathsf{d}}_T \bot$ by $\lnot$e. \qed
\end{proof}

\begin{lemma}[$\dnegr$]
  \label{lem:soundness_dneg}
  Let $(T,\Gamma,\Gamma')$ be a well-formed branch.
  Assume $\lnot \lnot s \in T,\Gamma$.
  If $\Gamma,s \vdash^{\mathsf{d}}_T \bot$
  then $\Gamma \vdash^{\mathsf{d}}_T \bot$.
\end{lemma}

\begin{proof}
  From $\Gamma,s \vdash^{\mathsf{d}}_T \bot$ we conclude
  $\Gamma \vdash^{\mathsf{d}}_T \lnot s$.
  We have $\Gamma \vdash^{\mathsf{d}}_T \lnot\lnot s$ and
  therefore $\Gamma \vdash^{\mathsf{d}}_T s$ by $\lnot\lnot$e.
  An application of $\lnot$e concludes the proof. \qed
\end{proof}

\begin{lemma}[$\impr$]
  \label{lem:soundness_imp}
  Let $(T,\Gamma,\Gamma')$ be a well-formed branch.
  Assume $s \limp t \in T,\Gamma$.
  If $\Gamma,\lnot s \vdash^{\mathsf{d}}_T \bot$
  and $\Gamma,t \vdash^{\mathsf{d}}_T \bot$
  then $\Gamma \vdash^{\mathsf{d}}_T \bot$.
\end{lemma}

\begin{proof}
  From the assumptions, we get $\Gamma \vdash^{\mathsf{d}}_T s$ and
  $\Gamma \vdash^{\mathsf{d}}_T \lnot t$.  Together with
  $\Gamma \vdash^{\mathsf{d}}_T s \limp t$, we obtain
  $\Gamma \vdash^{\mathsf{d}}_T t$ and therefore
  $\Gamma \vdash^{\mathsf{d}}_T \bot$. \qed
\end{proof}

\begin{lemma}[$\nimpr$]
  \label{lem:soundness_nimp}
  Let $(T,\Gamma,\Gamma')$ be a well-formed branch.
  Assume $\lnot (s \limp t) \in T,\Gamma$.
  If $\Gamma,s,\lnot t \vdash^{\mathsf{d}}_T \bot$
  then $\Gamma \vdash^{\mathsf{d}}_T \bot$.
\end{lemma}

\begin{proof}
  From the assumption, we get
  $\Gamma \vdash^{\mathsf{d}}_T s \limp t$.  Hence, an application of
  $\lnot$e yields $\Gamma \vdash^{\mathsf{d}}_T \bot$. \qed
\end{proof}

\begin{lemma}[$\allr$]
  \label{lem:soundness_all}
  Let $(T,\Gamma,\Gamma')$ be a well-formed branch.
  Assume $\forall x\colon A.s \in T,\Gamma$
  and $\Gamma \vdash^{\mathsf{d}}_T t\colon A$.
  If $\Gamma,[s[x/t]] \vdash^{\mathsf{d}}_T \bot$
  then $\Gamma \vdash^{\mathsf{d}}_T \bot$.
\end{lemma}

\begin{proof}
  From $\Gamma,[s[x/t]] \vdash^{\mathsf{d}}_T \bot$
  we get $\Gamma \vdash^{\mathsf{d}}_T \lnot [s[x/t]]$.
  We also have
  $\Gamma \vdash^{\mathsf{d}}_T \forall x\colon A.s$
  and therefore $\Gamma \vdash^{\mathsf{d}}_T [s[x/t]]$
  by $\forall$e and \lemref{bracket}. Finally,
  $\Gamma \vdash^{\mathsf{d}}_T \bot$ by $\lnot$e. \qed
\end{proof}

\begin{lemma}[$\nallr$]
  \label{lem:soundness_nall}
  Let $(T,\Gamma,\Gamma')$ be a well-formed branch.
  Assume $\lnot \forall x\colon A.s \in T,\Gamma$
  and $y$ is fresh for $T,\Gamma$.
  If $\Gamma,y\colon A,\lnot[s[x/y]] \vdash^{\mathsf{d}}_T \bot$
  then $\Gamma \vdash^{\mathsf{d}}_T \bot$.
\end{lemma}

\begin{proof}
  From $\Gamma,y\colon A,\lnot[s[x/y]] \vdash^{\mathsf{d}}_T \bot$ we
  obtain $\Gamma,y\colon A \vdash^{\mathsf{d}}_T s[x/y]$ using
  \lemref{bracket}. An application of $\forall$i yields
  $\Gamma \vdash^{\mathsf{d}}_T \forall y\colon A.s[x/y]$ which is the
  same as $\Gamma \vdash^{\mathsf{d}}_T \forall x\colon A.s$.  Since
  we also have
  $\Gamma \vdash^{\mathsf{d}}_T \lnot \forall x\colon A.s$, we
  conclude $\Gamma \vdash^{\mathsf{d}}_T \bot$ by $\lnot$e. \qed
\end{proof}

\begin{lemma}[$\ber$]
  \label{lem:soundness_be}
  Let $(T,\Gamma,\Gamma')$ be a well-formed branch.
  Assume $s \neq_o t \in T,\Gamma$.
  If $\Gamma,s,\lnot t \vdash^{\mathsf{d}}_T \bot$
  and $\Gamma,\lnot s,t \vdash^{\mathsf{d}}_T \bot$
  then $\Gamma \vdash^{\mathsf{d}}_T \bot$.
\end{lemma}

\begin{proof}
  From the assumptions we obtain
  $\Gamma,s \vdash^{\mathsf{d}}_T t$ and
  $\Gamma,\lnot s \vdash^{\mathsf{d}}_T \lnot t$
  which is equivalent to
  $\Gamma,t \vdash^{\mathsf{d}}_T s$.
  Hence, the rule
  \[\frac{\Gamma,s \vdash_T t \quad \Gamma,t \vdash_T s}{\Gamma \vdash_T s =_o t}\text{propExt}\]
  which is proven admissible in the appendix of \cite{RRB23ext}
    (the extended version of \cite{RRB23}) yields
  $\Gamma \vdash^{\mathsf{d}}_T s =_o t$ and we obtain
  $\Gamma \vdash^{\mathsf{d}}_T \bot$ with $\lnot$e. \qed
\end{proof}

\begin{lemma}[$\bqr$]
  \label{lem:soundness_bq}
  Let $(T,\Gamma,\Gamma')$ be a well-formed branch.
  Assume $s =_o t \in T,\Gamma$.
  If $\Gamma,s,t \vdash^{\mathsf{d}}_T \bot$
  and $\Gamma,\lnot s,\lnot t \vdash^{\mathsf{d}}_T \bot$
  then $\Gamma \vdash^{\mathsf{d}}_T \bot$.
\end{lemma}

\begin{proof}
  By assumption, we have $\Gamma, s \vdash^{\mathsf{d}}_T s =_o t$ and
  $\Gamma, s \vdash^{\mathsf{d}}_T s$.  Hence, we obtain
  $\Gamma, s \vdash^{\mathsf{d}}_T t$ by cong$\vdash$.  Furthermore,
  $\lnot$i applied to $\Gamma,s,t \vdash^{\mathsf{d}}_T \bot$ yields
  $\Gamma,s \vdash^{\mathsf{d}}_T \lnot t$.  Hence, we arrive at
  $\Gamma,s \vdash^{\mathsf{d}}_T \bot$ by an application of $\lnot$e.
  In a similar way, we can derive
  $\Gamma, \lnot s \vdash^{\mathsf{d}}_T \bot$.  Therefore, an
  application of \lemref{xmcases} gives us the desired result. \qed
\end{proof}

\begin{lemma}[$\matr$]
  \label{lem:soundness_mat}
  Let $(T,\Gamma,\Gamma')$ be a well-formed branch.
  Assume $x\: s_1 \:\dots\: s_n \in T,\Gamma$,
  $\lnot x\: t_1 \:\dots\: t_n \in T, \Gamma$
  and let $\Gamma \vdash^{\mathsf{d}}_T x\colon \dpi y_1\colon A_1 \cdots \dpi y_n \colon A_n.o$.
  If
  \[\Gamma,s_i \neq_{A_i[x_1/s_1,\dots,x_{i-1}/s_{i-1}]} t_i \vdash^{\mathsf{d}}_T \bot\]
  for all $1 \leq i \leq n$ then $\Gamma \vdash^{\mathsf{d}}_T \bot$.
\end{lemma}

\begin{proof}
  Note that $x s_1 \dots s_n \in T,\Gamma$ and
  $\lnot x\: t_1 \:\dots\: t_n \in T,\Gamma$ implies
  \[\Gamma \vdash^{\mathsf{d}}_T x\: s_1 \:\dots\: s_n \neq_o x\: t_1 \:\dots\: t_n.\]
  Since in \lemref{soundness_dec} it makes no difference whether the assumption is
  contained in $T,\Gamma$ or is merely provable from it, \lemref{soundness_dec} can be
  used to complete the proof. \qed
\end{proof}

\begin{lemma}[$\conr$]
  \label{lem:soundness_con}
  Let $(T,\Gamma,\Gamma')$ be a well-formed branch.
  Assume $s=_{a t_1 \dots t_n} t, u\not=_{a t_1 \dots t_n} v\in T,\Gamma$.
  If
  $\Gamma,s\not=u,t\not=u \vdash^{\mathsf{d}}_T \bot$
  and
  $\Gamma,s\not=v,t\not=v \vdash^{\mathsf{d}}_T \bot$
  then $\Gamma \vdash^{\mathsf{d}}_T \bot$.
\end{lemma}

\begin{proof}
  Let
  $\Gamma,s\not=u,t\not=u\vdash^{\mathsf{d}}_T \bot$
  and
  $\Gamma,s\not=v,t\not=v\vdash^{\mathsf{d}}_T \bot$.
  By Lemma~\ref{lem:xmcases} it is enough to prove
  $\Gamma,s= u\vdash^{\mathsf{d}}_T \bot$
  and
  $\Gamma,s\not= u\vdash^{\mathsf{d}}_T \bot$.

  We first prove $\Gamma,s= u\vdash^{\mathsf{d}}_T \bot$.
  By Lemma~\ref{lem:xmcases} it is enough to prove
  $\Gamma,s= u,t= v\vdash^{\mathsf{d}}_T \bot$
  and
  $\Gamma,s= u,t \not= v\vdash^{\mathsf{d}}_T \bot$.
  Using $s= t\in T,\Gamma$ and symmetry and transitivity
  (which is an admissible rule \cite{RRB23ext})
  we can infer $\Gamma,s= u,t= v\vdash^{\mathsf{d}}_T u = v$.
  This along with $u\not= v\in T,\Gamma$ gives
  $\Gamma,s= u,t= v\vdash^{\mathsf{d}}_T \bot$.
  In order to prove $\Gamma,s= u,t \not= v\vdash^{\mathsf{d}}_T \bot$
  we again use Lemma~\ref{lem:xmcases} to reduce
  to proving
  $\Gamma,s= u,t \not= v, s = v\vdash^{\mathsf{d}}_T \bot$
  and
  $\Gamma,s= u,t \not= v, s \not= v\vdash^{\mathsf{d}}_T \bot$.
  We know $\Gamma,s= u,t \not= v, s \not= v\vdash^{\mathsf{d}}_T \bot$
  already since $\Gamma,t \not= v, s \not= v\vdash^{\mathsf{d}}_T \bot$.
  By symmetry and transitivity,
  $\Gamma,s= u,t \not= v, s = v\vdash^{\mathsf{d}}_T u = v$
  giving 
  $\Gamma,s= u,t \not= v, s = v\vdash^{\mathsf{d}}_T \bot$.
  Hence $\Gamma,s= u,t \not= v\vdash^{\mathsf{d}}_T \bot$
  and $\Gamma,s= u\vdash^{\mathsf{d}}_T \bot$.

  All that remains is to prove $\Gamma,s\not= u\vdash^{\mathsf{d}}_T \bot$.
  By Lemma~\ref{lem:xmcases} it is enough to prove
  $\Gamma,s\not= u,t = u\vdash^{\mathsf{d}}_T \bot$
  since we already know $\Gamma,s\not= u,t \not= u\vdash^{\mathsf{d}}_T \bot$.
  By Lemma~\ref{lem:xmcases} it is enough to prove
  $\Gamma,s\not= u,t = u,t = v\vdash^{\mathsf{d}}_T \bot$
  and
  $\Gamma,s\not= u,t = u,t \not= v\vdash^{\mathsf{d}}_T \bot$.
  We have $\Gamma,s\not= u,t = u,t = v\vdash^{\mathsf{d}}_T \bot$
  since $\Gamma,s\not= u,t = u,t = v\vdash^{\mathsf{d}}_T u = v$
  by symmetry and transitivity.
  By Lemma~\ref{lem:xmcases} all that remains is to prove
  $\Gamma,s\not= u,t = u,t \not= v, s = v\vdash^{\mathsf{d}}_T \bot$
  and
  $\Gamma,s\not= u,t = u,t \not= v, s \not= v\vdash^{\mathsf{d}}_T \bot$.
  We know $\Gamma,s\not= u,t = u,t \not= v, s \not= v\vdash^{\mathsf{d}}_T \bot$
  since $\Gamma,t \not= v, s \not= v\vdash^{\mathsf{d}}_T \bot$.
  We finally have $\Gamma,s\not= u,t = u,t \not= v, s = v\vdash^{\mathsf{d}}_T \bot$
  as a consequence of
  $\Gamma,s\not= u,t = u,t \not= v, s = v\vdash^{\mathsf{d}}_T u = v$,
  which again follows from symmetry and transitivity. \qed
\end{proof}

\end{document}